\documentclass[11pt,a4paper]{article}
\usepackage[T1]{fontenc}
\usepackage[utf8]{inputenc}
\usepackage{lmodern}
\usepackage{amsmath,amssymb,amsthm,mathtools,mathrsfs}
\usepackage{microtype}
\usepackage{geometry}
\usepackage{booktabs}
\usepackage{enumitem}
\usepackage{xcolor}
\usepackage{tikz}
\usetikzlibrary{arrows.meta,calc,decorations.pathmorphing,positioning}
\usepackage{hyperref}
\definecolor{deepblue}{RGB}{28,52,84}
\definecolor{midblue}{RGB}{54,91,135}
\definecolor{burgundy}{RGB}{128,36,52}
\definecolor{softgray}{RGB}{245,246,248}

\hypersetup{
	colorlinks=true,
	linkcolor=blue!55!black,
	citecolor=blue!55!black,
	urlcolor=blue!55!black,
	pdftitle={A generalization of the Fredenhagen-Haag derivation of Hawking radiation for a class of Vaidya space-times.},
	pdfauthor={Felipe Dilho Alves}
}

\newtheorem{theorem}{Theorem}[section]
\newtheorem{proposition}[theorem]{Proposition}
\newtheorem{lemma}[theorem]{Lemma}
\newtheorem{corollary}[theorem]{Corollary}
\theoremstyle{definition}
\newtheorem{definition}[theorem]{Definition}
\newtheorem{assumption}[theorem]{Assumption}
\theoremstyle{remark}

\newcommand{\M}{\mathcal M}
\newcommand{\D}{\mathcal D}
\newcommand{\Ocal}{\mathcal O}
\newcommand{\scrI}{\mathscr I}
\newcommand{\HH}{\mathcal H}
\newcommand{\Cinf}{C_0^\infty}
\newcommand{\eps}{\varepsilon}
\newcommand{\dd}{\mathrm d}
\newcommand{\supp}{\operatorname{supp}}
\newcommand{\WF}{\operatorname{WF}}

\title{A generalization of the Fredenhagen-Haag derivation of Hawking radiation for a class of Vaidya space-times.}
\author{Felipe Dilho Alves}
\date{\today}

\begin{document}
	\maketitle
	
	\begin{abstract}
		We develop a quantitative Fredenhagen--Haag like approach for describing Hawking radiation using massless scalar
		fields on controlled spherically symmetric Vaidya space-times.  The local
		thermal character of the system is supplied by the universal scaling limit of Hadamard
		two-point functions at an outer trapping horizon
		\cite{KurpiczPinamontiVerch2021}.  On regular detector--horizon windows, we
		construct globally hyperbolic developments and compare the nonautonomous
		Vaidya evolution with a frozen Schwarzschild propagator on Sobolev energy
		spaces.  We derive an explicit Duhamel estimate that is uniform in angular
		momentum, calculate the exact linear null-peeling coefficient, and bound the
		quadratic remainder of the ray map.  Combining these estimates with
		positivity gives a two-sided detector-response inequality relative to the
		local thermal reference.  Its error terms quantify operator variation,
		stationary scattering tails, finite Hadamard scaling, horizon localisation,
		and the outgoing channel.  The inequality is valid at finite parameters
		because each of these contributions is retained.  Under the decay hypotheses, the detector response converges to the corresponding
		Fredenhagen--Haag form for asymptotically stationary accretion and for
		evaporation--accretion turnaround profiles.  We also construct a Hadamard
		state by Cauchy transport from an eventually stationary Unruh covariance and
		show that a finite evaporating slab does not determine a late-time response
		without a prescribed future extension.  For asymptotic evaporation with
		$m(u)>0$ at every finite time and $m(u)\to0$, a mass-rescaled conformal
		formulation yields a scale-covariant finite-window estimate for
		scale-following detectors.
	\end{abstract}
	
	\tableofcontents
	
	\section{Introduction}
	
	Fredenhagen and Haag related the response of a detector at late Schwarzschild
	time to the short-distance behaviour of a quantum state near the sphere at
	which a collapsing star crosses its Schwarzschild radius
	\cite{FredenhagenHaag1990}.  Their argument has two conceptually distinct
	parts.  First, a classical solution obtained by propagating the detector
	smearing function backwards decomposes into an asymptotic packet and a packet
	which is exponentially compressed towards the horizon.  Second, the
	universal leading singularity of every Hadamard two-point function converts
	that compression into a Planck factor.  The first part is global and uses
	stationary Schwarzschild scattering; the second is local and
	state-independent at leading order.
	
	In a Vaidya space-time the Vaidya mass function $m(w)$ depends on the null
	coordinate $w$.  It equals the Misner--Sharp mass and, in the outgoing
	asymptotically flat orientation, the Bondi mass.  When the completed
	space-time contains a black hole, $m(w)$ is the evolving mass parameter of
	the black-hole geometry.  The Schwarzschild time translation is therefore no
	longer an isometry.  Consequently, separation of variables
	does not reduce the radial field equation to a time-independent one-dimensional
	scattering problem, frequency is not conserved, and a greybody coefficient
	$D_\ell(\omega)$ must in general be replaced by a frequency-mixing operator.
	It is therefore not enough to replace the Schwarzschild mass by a function in
	the final formulas of \cite{FredenhagenHaag1990}.
	
	The local component of the desired generalisation is supplied by an existing
	horizon scaling theorem.
	For spherically symmetric space-times with an outer trapping horizon,
	Kurpicz, Pinamonti and Verch proved that the horizon scaling limit of a
	Hadamard two-point function is universal and thermal with respect to the
	projected Kodama flow, with inverse temperature $2\pi/\kappa$
	\cite{KurpiczPinamontiVerch2021}.  Related characteristic constructions of
	Hadamard states and local observables have been developed in
	\cite{JanssenVerch2023}.  These results do not, by themselves, identify the
	response of a detector in the asymptotic region.  That identification requires
	controlled propagation from the detector to the horizon.
	
	The purpose of this paper is to obtain a quantitative propagation result from
	scattering and energy estimates and to combine it with the local horizon
	scaling theorem.  The conclusion is expressed by the bound
	\begin{equation}
		\max\{0,\mathcal F_{\mathrm{fr}}-\mathcal E\}
		\leq \mathcal F_g
		\leq \mathcal F_{\mathrm{fr}}+\mathcal E,
		\label{eq:intro-bound}
	\end{equation}
	where $\mathcal F_g$ is the response in the dynamical geometry,
	$\mathcal F_{\mathrm{fr}}$ is an exactly defined frozen reference response,
	and $\mathcal E$ is an explicit positive error functional.  Section
	\ref{sec:reduction} proves the abstract positivity inequality,
	Section~\ref{sec:target-propagation} proves the Vaidya PDE comparison, and
	Theorem~\ref{thm:certified-response} records precisely the additional bridge
	needed to combine it with local horizon scaling.
	
	Throughout, the metric has signature $(-,+,+,+)$ and
	$G=c=\hbar=k_{\mathrm B}=1$.
	
	\section{The Fredenhagen--Haag mechanism}
	\label{sec:FH}
	
	We briefly separate the ingredients of the Fredenhagen--Haag argument that
	survive in a dynamical geometry from those that depend on stationarity.
	Assume that $(\M,g)$ is globally hyperbolic, and let $\Phi$ be a neutral
	scalar quantum field satisfying, in the distributional sense, the massless
	covariant wave equation
	\begin{equation}
		\Box_g\Phi=0,
		\qquad
		\Box_g
		=
		|g|^{-1/2}\partial_\mu
		\left(
		|g|^{1/2}g^{\mu\nu}\partial_\nu
		\right),
		\qquad
		|g|=-\det(g_{\mu\nu}).
		\label{eq:massless-wave-equation}
	\end{equation}
	For $h\in\Cinf(\M)$ supported in a distant observation region, define the
	smeared field
	\begin{equation}
		Q=\Phi(h)
		=\int_{\M}\Phi(x)h(x)\,\dd\mathrm{vol}_g(x).
		\label{eq:smeared-detector}
	\end{equation}
	
	Global hyperbolicity guarantees the existence and uniqueness of the retarded
	and advanced Green operators $E_{\mathrm{ret}}$ and $E_{\mathrm{adv}}$ of
	$\Box_g$.  Define the causal propagator and its associated classical solution
	by
	\begin{equation}
		E:=E_{\mathrm{ret}}-E_{\mathrm{adv}},
		\qquad
		f:=Eh.
		\label{eq:causal-solution}
	\end{equation}
	Then $f\in C^\infty_{\mathrm{sc}}(\M)$ is smooth, spacelike compact, and
	satisfies
	\[
	\Box_g f=0.
	\]
	The corresponding symplectic current is
	\begin{equation}
		J_a:=\Phi\nabla_a f-f\nabla_a\Phi.
		\label{eq:symplectic-current}
	\end{equation}
	Using the field equations, its divergence vanishes:
	\begin{equation}
		\nabla^aJ_a
		=
		\Phi\Box_gf-f\Box_g\Phi
		=0.
		\label{eq:symplectic-current-conservation}
	\end{equation}
	Stokes's theorem therefore gives
	\begin{equation}
		\int_{\Sigma_1}J_a n_1^a\,\dd\Sigma_1
		=
		\int_{\Sigma_2}J_a n_2^a\,\dd\Sigma_2
		\label{eq:symplectic-conservation}
	\end{equation}
	for any two Cauchy surfaces $\Sigma_1$ and $\Sigma_2$.
	
	Conservation alone shows only that the surface integral is independent of the
	chosen Cauchy surface.  Its identification with the smeared field follows
	from the particular choice $f=Eh$.  Indeed, Green's identity gives
	\begin{equation}
		Q=\Phi(h)
		=
		\int_\Sigma
		\left(
		\Phi\nabla_a f-f\nabla_a\Phi
		\right)n^a\,\dd\Sigma
		\label{eq:symplectic-representation}
	\end{equation}
	for every Cauchy surface $\Sigma$, with the orientations of $\Sigma$ and the
	sign convention in \eqref{eq:causal-solution} chosen consistently.  The
	contraction
	\[
	n^a\nabla_a
	\]
	is the derivative normal to $\Sigma$ and hence contains the appropriate time
	derivative.  This construction requires a well-posed hyperbolic problem, but
	does not require a Killing vector field.
	
	Figure~\ref{fig:FH-collapse-detector} depicts the causal content of this
	time-slice step.  The detector remains in the exterior, while its associated
	classical solution is represented by data on an earlier Cauchy surface.
	
	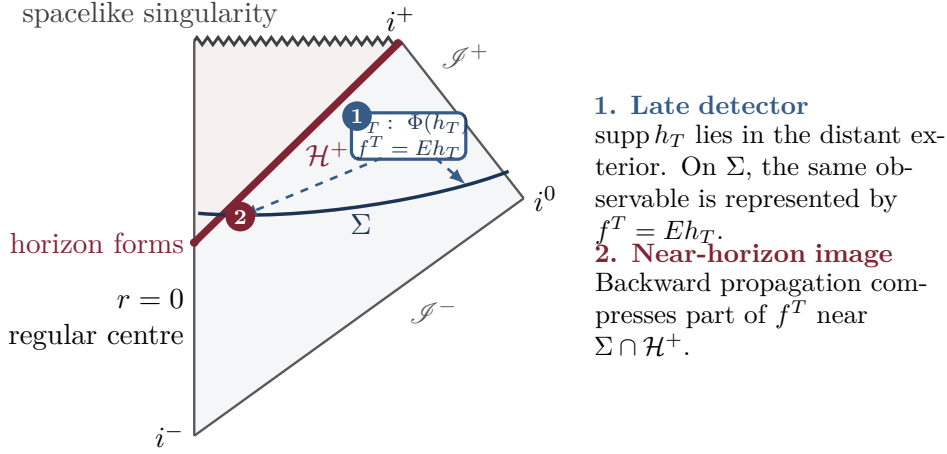
\begin{figure}[htbp]
		\centering
			\begin{tikzpicture}[x=1.05cm,y=0.73cm,
				boundary/.style={black!65,line width=0.9pt},
				backray/.style={-{Latex[length=2.2mm]},midblue,line width=1pt,dashed}]
				\coordinate (iminus) at (0,-1.25);
				\coordinate (izero) at (4.15,3.05);
				\coordinate (iplus) at (2.6,5.9);
				\coordinate (sing0) at (0,5.9);
				\coordinate (form) at (0,2.25);
				\fill[softgray] (iminus)--(izero)--(iplus)--(sing0)--cycle;
				\fill[burgundy!7] (form)--(iplus)--(sing0)--cycle;
				\draw[boundary] (iminus)--(izero) node[pos=.62,below right] {$\mathscr I^-$};
				\draw[boundary] (izero)--(iplus) node[pos=.76,above right] {$\mathscr I^+$};
				\draw[boundary] (iminus)--(sing0);
				\draw[decorate,decoration={zigzag,segment length=4pt,amplitude=1.4pt},
				black!75,line width=1.1pt] (sing0)--(iplus)
				node[pos=.47,above left=1pt] {spacelike singularity};
				\draw[burgundy,line width=2.7pt] (form)--(iplus)
				node[pos=.52,below right=-5pt,text=burgundy] {$\mathcal H^+$};
				\node[left] at (iminus) {$i^-$};
				\node[right] at (izero) {$i^0$};
				\node[above] at (iplus) {$i^+$};
				\node[left,align=right] at (0,0.85) {$r=0$\\regular centre};
				\fill[burgundy] (form) circle (1.6pt);
				\node[left,text=burgundy] at (form) {horizon forms};
				
				\draw[deepblue,line width=1.3pt]
				(0.05,2.78) .. controls (1.38,2.62) and (3.10,3.06) .. (3.96,3.54);
				\node[below,text=deepblue,fill=softgray,inner sep=1pt] at (2.10,2.82) {$\Sigma$};
				
				\coordinate (det) at (2.70,4.18);
				\filldraw[fill=white,draw=midblue,line width=1.2pt,rounded corners=3pt]
				($(det)+(-.72,-.42)$) rectangle ($(det)+(.72,.42)$);
				\node[align=center,text=deepblue,font=\scriptsize] at (det)
				{$\mathcal O_T:\ \Phi(h_T)$\\$f^T=Eh_T$};
				\coordinate (nearhor) at (0.58,2.72);
				\coordinate (outdata) at (3.42,3.30);
				\draw[backray] ($(det)+(-.34,-.40)$) -- (nearhor);
				\draw[backray] ($(det)+(.34,-.40)$) -- (outdata);
				\node[circle,fill=midblue,text=white,font=\bfseries\scriptsize,
				inner sep=1.4pt] at (2.05,4.54) {1};
				\node[circle,fill=burgundy,text=white,font=\bfseries\scriptsize,
				inner sep=1.4pt] at (nearhor) {2};
				
				\node[align=left,text width=4.90cm,anchor=north west,font=\small]
				at (4.92,5.10)
				{\textcolor{midblue}{\bfseries 1. Late detector}\\[-1pt]
					$\operatorname{supp}h_T$ lies in the distant exterior.  On $\Sigma$, the same observable is represented by $f^T=Eh_T$.};
				\node[align=left,text width=4.65cm,anchor=north west,font=\small]
				at (4.92,2.42)
				{\textcolor{burgundy}{\bfseries 2. Near-horizon image}\\[-1pt]
					Backward propagation compresses part of $f^T$ near $\Sigma\cap\mathcal H^+$.};
			\end{tikzpicture}
		\caption{One-sided collapse diagram for the Fredenhagen--Haag time-slice
			construction.}
		\label{fig:FH-collapse-detector}
	\end{figure}
	
	In the stationary Schwarzschild exterior considered by Fredenhagen and Haag,
	let $h_T$ be the Schwarzschild-time translate of $h$ and put
	\[
	Q_T=\Phi(h_T),
	\qquad
	f^T=Eh_T.
	\]
	Their scattering analysis shows that the restriction of $f^T$ to the early
	surface $\tau=0$ admits the asymptotic decomposition
	\begin{equation}
		f^T=f_+^T+f_-^T+\Delta^T,
		\qquad
		\Delta^T\longrightarrow0
		\quad\text{as }T\longrightarrow+\infty,
		\label{eq:FH-decomposition}
	\end{equation}
	in the topology specified by their propagation estimates.  The packet
	$f_+^T$ moves towards spatial infinity, whereas $f_-^T$ is exponentially
	compressed towards the horizon.  More precisely, if
	\[
	r_0=2M,
	\qquad
	\xi=\frac{r-r_0}{r_0},
	\qquad
	\kappa=\frac{1}{2r_0}=\frac{1}{4M},
	\]
	then Fredenhagen and Haag give the following leading near-horizon expression,
	after suppressing factors which tend to unity as $r\to r_0$:
	\begin{equation}
		f_-^T(\tau,r)
		=
		\psi\left(
		\frac{\xi}{\lambda_T}e^{-\kappa\tau}
		\right),
		\qquad
		\lambda_T=e^{-\kappa T}.
		\label{eq:FH-dilation}
	\end{equation}
	In particular, on $\tau=0$,
	\[
	f_-^T(0,r)
	=
	\psi\left(\frac{\xi}{\lambda_T}\right).
	\]
	Thus the part displayed explicitly by Fredenhagen and Haag shows that a late
	Schwarzschild-time translation corresponds, near the horizon, to a dilation
	towards $\xi=0$, with scaling parameter
	$\lambda_T=e^{-\kappa T}$.
	
	The three terms in \eqref{eq:FH-decomposition} are represented schematically
	in Figure~\ref{fig:FH-three-channels}.  The wave shapes indicate localisation
	only; the decay of $\Delta^T$ is the content of the scattering estimate, not
	an assumption inferred from the drawing.
	
	\begin{figure}[htbp]
		\centering
		\begin{tikzpicture}[x=1cm,y=0.80cm,
			flow/.style={-{Latex[length=2.3mm]},line width=1.1pt}]
			\fill[softgray] (0.7,0.55) rectangle (13.1,5.8);
			\draw[deepblue,line width=1.2pt] (0.75,0.95)--(11.15,0.95)
			node[right,text=deepblue] {early $\Sigma$};
			\draw[burgundy,line width=2.4pt] (1.05,0.58)--(1.05,4.95)
			node[above,text=burgundy] {$\mathcal H^+$};
			\draw[black!55,-{Latex[length=2mm]}] (1.32,0.72)--(10.85,0.72)
			node[right] {$r_*$};
			
			\draw[black,line width=1.7pt,domain=0:2.2,samples=100,variable=\x]
			plot ({6.05+\x},{4.28+0.33*exp(-2.4*(\x-1.1)^2)*sin(720*\x)});
			\node[above,text=black] at (7.15,4.62) {late packet $f^T=Eh_T$};
			\draw[flow,burgundy] (6.65,3.92) .. controls (5.8,3.05) and (3.0,2.15) .. (1.75,1.42);
			\draw[flow,midblue] (7.58,3.92) .. controls (8.1,3.00) and (8.7,2.05) .. (9.05,1.42);
			\draw[flow,black!38,dashed] (7.12,3.90) -- (6.12,1.48);
			\draw[burgundy,line width=1.55pt,domain=0:1.35,samples=150,variable=\x]
			plot ({1.22+\x},{1.25+0.27*exp(-3.1*(\x-.52)^2)*sin(1500*\x)});
			\node[align=center,text=burgundy, rounded corners=2pt,inner sep=2pt]
			at (2.66,3.0) {\hspace*{-36pt}$f_-^T$ \\ compressed at $\mathcal H^+$};
			\draw[midblue,line width=1.55pt,domain=0:2.0,samples=110,variable=\x]
			plot ({8.18+\x},{1.25+0.27*exp(-2.5*(\x-1)^2)*sin(650*\x)});
			\node[align=center,text=deepblue, rounded corners=2pt,inner sep=2pt]
			at (9.20,2.40) {\hspace*{-4pt} $f_+^T$\\ \hspace*{85pt} towards spatial infinity};
			\foreach \k/\opa in {0/0.48,1/0.31,2/0.16}{
				\draw[black!55,opacity=\opa,line width=1.15pt,decorate,
				decoration={snake,amplitude=1.5pt,segment length=7pt}]
				(5.50+0.35*\k,1.20+0.18*\k)--(6.55-0.18*\k,1.20+0.18*\k);
			}
			\node[align=center,text=black!55, rounded corners=2pt,inner sep=2pt]
			at (6.05,1.90) {$\Delta^T$\\ \hspace*{-57pt}$\|\Delta^T\|\to0$};
		\end{tikzpicture}
		\caption{Schematic backward decomposition of the late detector packet into
			the horizon channel, the outgoing channel and the vanishing stationary
			remainder.}
		\label{fig:FH-three-channels}
	\end{figure}
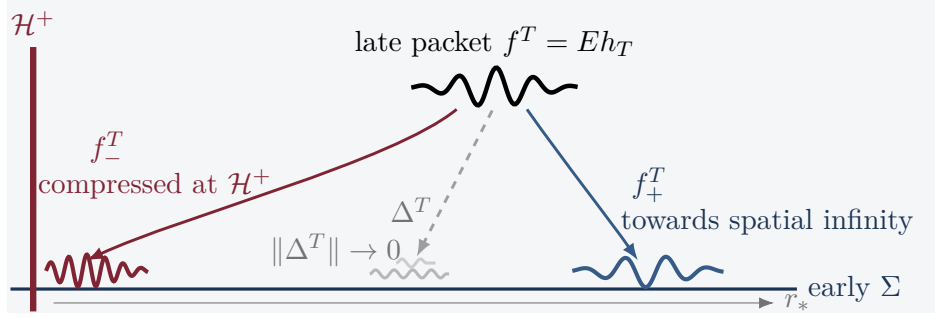
	
	After transferring the normal derivatives to the two-point function and taking the radial distributional scaling limit, the
	leading Hadamard singularity becomes supported on the angular diagonal
	$\Omega_1=\Omega_2$.  After integration over the angular variables, or
	equivalently after projection onto spherical harmonics, its remaining radial
	part is proportional to
	\begin{equation}
		K(\xi_1,\xi_2)
		:=
		\lim_{\varepsilon\downarrow0}
		\frac{1}{
			(\xi_1-\xi_2+i\varepsilon)^2
		},
		\qquad
		\text{with the limit taken in }\D'(\mathbb R^2).\footnote{This expression is the precise distributional meaning sometimes abbreviated by $(\xi_1-\xi_2+i0)^{-2}$.}
		\label{eq:horizon-kernel}
	\end{equation}
	
	On the exterior side of the horizon, where $\xi>0$, introduce the logarithmic coordinate
	\[
	\xi=e^{\kappa y},
	\]
	and define
	\[
	\phi(y):=\psi(e^{\kappa y}).
	\]
	
	Since
	\[
	\dd\xi=\kappa e^{\kappa y}\,\dd y
	\]
	and
	\[
	e^{\kappa y_1}-e^{\kappa y_2}
	=
	2e^{\kappa(y_1+y_2)/2}
	\sinh\left(\frac{\kappa}{2}(y_1-y_2)\right),
	\]
	the quadratic form determined by \eqref{eq:horizon-kernel} becomes,
	up to the overall normalisation inherited from the Hadamard kernel,
	\begin{equation}
		\frac{\kappa^2}{4}
		\int_{\mathbb R^2}
		S_\kappa(y_1-y_2)\,
		\overline{\phi(y_1)}\phi(y_2)\,
		\dd y_1\dd y_2,
		\label{eq:logarithmic-horizon-form}
	\end{equation}
	where
	\begin{equation}
		S_\kappa(s)
		:=
		\lim_{\varepsilon\downarrow0}
		\left[
		\sinh\left(\frac{\kappa s}{2}+i\varepsilon\right)
		\right]^{-2}
		\quad\text{in }\D'(\mathbb R).
		\label{eq:sinh-kernel}
	\end{equation}
	The kernel $S_\kappa$ depends only on $y_1-y_2$ and is therefore
	translation invariant in the logarithmic horizon coordinate.
	
	With the chosen Fourier-transform convention, the Fourier transform of
	$S_\kappa$ is proportional to the nonsymmetrised bosonic KMS two-point
	spectral density
	\begin{equation}
		\frac{\omega}{1-e^{-\beta\omega}},
		\qquad
		\beta=\frac{2\pi}{\kappa}=8\pi M.
		\label{eq:KMS-factor}
	\end{equation}
	The proportionality constant depends on the normalisation and
	Fourier-transform convention.
	Stationary radial scattering then multiplies the transmitted part of the
	response by the greybody factor $|D_\ell(\omega)|^2$.
	
	The dynamical problem is therefore not to derive
	\eqref{eq:horizon-kernel} again: that distribution is fixed by the universal
	leading Hadamard singularity.  The genuinely nonstationary problem is to
	replace the Schwarzschild decomposition
	\eqref{eq:FH-decomposition} and the exponential relation
	\eqref{eq:FH-dilation} by corresponding statements, with controlled
	remainders, for the nonautonomous Vaidya wave equation.
	
	\section{Vaidya geometry and causal hypotheses}
	\label{sec:geometry}
	
	\subsection{Advanced and retarded forms}
	
	It is useful to write both Vaidya forms as
	\begin{equation}
		g_{\epsilon}
		=-C(w,r)\,\dd w^2+2\epsilon\,\dd w\,\dd r+r^2\dd\Omega^2,
		\qquad
		C(w,r)=1-\frac{2m(w)}{r},
		\qquad \epsilon\in\{+1,-1\}.
		\label{eq:unified-Vaidya}
	\end{equation}
	For $\epsilon=+1$, $w=v$ is an advanced coordinate and
	\begin{equation}
		g_+=-\left(1-\frac{2m(v)}r\right)\dd v^2
		+2\,\dd v\,\dd r+r^2\dd\Omega^2.
		\label{eq:ingoing-Vaidya}
	\end{equation}
	For $\epsilon=-1$, $w=u$ is a retarded coordinate and
	\begin{equation}
		g_-=-\left(1-\frac{2m(u)}r\right)\dd u^2
		-2\,\dd u\,\dd r+r^2\dd\Omega^2.
		\label{eq:outgoing-Vaidya}
	\end{equation}
	The Einstein tensor corresponds to a null fluid with
	\begin{equation}
		T_{ww}=\frac{\epsilon m'(w)}{4\pi r^2}.
		\label{eq:Vaidya-stress}
	\end{equation}
	Thus positive-energy ingoing null dust requires \(m'(v)\geq 0\), whereas positive-energy outgoing null dust requires \(m'(u)\leq 0\). Conversely, a decreasing mass in the advanced metric, \(m'(v) < 0 \), represents an effective negative-energy flux in the ingoing null channel and is often used to model the near-horizon component of semiclassical evaporation. This flux is not locally identical to the positive-energy outgoing flux measured near future null infinity: the two occupy different null channels and correspond respectively to \( T_{vv} < 0 \) and \( T_{uu} > 0 \). A single Vaidya patch represents only one of these pure null components. A complete evaporation model therefore generally requires matched ingoing and outgoing patches or a more general double-null geometry, as in the classical constructions initiated in \cite{Hiscock1981,Hiscock1981II}.
	
	\subsection{A temporal function and stable causality}
	
	Recall that a space-time is \emph{causal} if it contains no closed causal
	curve.  It is \emph{stably causal} if its light cones can be widened slightly
	without producing a closed causal curve.  More precisely, $(\M,g)$ is stably
	causal if there exists a Lorentzian metric $\widetilde g$ whose causal cones
	are strictly wider than those of $g$ and for which $(\M,\widetilde g)$ is
	causal.
	
	A smooth function $\tau:\M\to\mathbb R$ is called a \emph{temporal function}
	if its gradient is everywhere timelike:
	\[
	g(\nabla\tau,\nabla\tau)
	=
	g^{-1}(\dd\tau,\dd\tau)<0.
	\]
	After replacing $\tau$ by $-\tau$ if necessary, its gradient may be taken to
	be past-directed.  In that convention, $\tau$ is strictly increasing along
	every future-directed causal curve.  The existence of a temporal function
	implies stable causality.
	
	The Vaidya metrics considered here possess a particularly simple temporal
	function.
	
	\begin{proposition}
		Let $m$ be smooth and nonnegative on the coordinate interval under
		consideration, and restrict to $r>0$.  For the metric
		\eqref{eq:unified-Vaidya}, the function
		\begin{equation}
			\tau_\epsilon=w-\epsilon r
			\label{eq:temporal-function}
		\end{equation}
		has an everywhere timelike gradient.  Consequently, every Vaidya coordinate
		region satisfying these assumptions is stably causal.
	\end{proposition}
	
	\begin{proof}
		The inverse of the $(w,r)$ part of the metric
		\eqref{eq:unified-Vaidya} is
		\begin{equation}
			g^{ww}=0,
			\qquad
			g^{wr}=g^{rw}=\epsilon,
			\qquad
			g^{rr}=C(w,r).
		\end{equation}
		Since
		\[
		\dd\tau_\epsilon=\dd w-\epsilon\,\dd r,
		\]
		we obtain
		\begin{align}
			g^{-1}(\dd\tau_\epsilon,\dd\tau_\epsilon)
			&=
			g^{ww}
			-2\epsilon g^{wr}
			+\epsilon^2g^{rr}\notag\\
			&=
			-2+C(w,r)\notag\\
			&=
			-1-\frac{2m(w)}r.
			\label{eq:temporal-calculation}
		\end{align}
		Because $m(w)\geq0$ and $r>0$,
		\[
		g^{-1}(\dd\tau_\epsilon,\dd\tau_\epsilon)
		\leq-1<0.
		\]
		Thus $\dd\tau_\epsilon$, or equivalently
		$\nabla\tau_\epsilon$, is everywhere timelike, so
		$\tau_\epsilon$ is a temporal function up to an overall choice of sign.
		
		The condition that a covector be timelike is open.  The light cones of $g$
		can therefore be widened slightly while keeping
		$\dd\tau_\epsilon$ timelike.  The function $\tau_\epsilon$ remains strictly
		monotone along every causal curve of the widened metric, and hence such a
		curve cannot be closed.  This proves stable causality.
	\end{proof}
	
	This result provides an important causal property of the Vaidya coordinate
	region: in particular, it excludes closed causal curves and implies strong
	causality.  It does not, however, establish global hyperbolicity.  A temporal
	function need not be a \emph{Cauchy temporal function}.  For the latter, every
	inextendible causal curve must intersect each of its level surfaces exactly
	once.
	
	Accordingly, global hyperbolicity of the space-time region used below must be
	verified separately.  One may do this either by constructing a Cauchy
	temporal function or by proving that the region is causal and that every
	causal diamond
	\[
	J^+(p)\cap J^-(q)
	\]
	is compact.  This additional global analysis depends on the chosen
	space-time domain, including its treatment of the central singularity,
	matching surfaces and any artificial boundaries.  The calculation
	\eqref{eq:temporal-calculation} establishes stable causality throughout the
	regular Vaidya region, but it does not by itself control those global
	boundaries.
	
	\subsection{Trapping and event horizons}
	
	For the advanced metric, choose the future-directed radial null fields
	\begin{equation}
		\ell=\partial_v+\frac{C}{2}\partial_r,
		\qquad n=-\partial_r,
		\qquad g(\ell,n)=-1.
	\end{equation}
	Their spherical expansions are
	\begin{equation}
		\theta_{(\ell)}=\frac{C}{r},
		\qquad
		\theta_{(n)}=-\frac{2}{r}.
	\end{equation}
	The spherical future trapping horizon is therefore
	\begin{equation}
		\HH_{\mathrm{tr}}=\{r=2m(v)\}.
		\label{eq:trapping-horizon}
	\end{equation}
	A tangent vector $X=\partial_v+2m'(v)\partial_r$ has, on
	$\HH_{\mathrm{tr}}$,
	\begin{equation}
		g(X,X)=4m'(v).
		\label{eq:horizon-signature}
	\end{equation}
	It follows that the horizon is spacelike during positive-energy accretion,
	null in the stationary case, and timelike for an effective negative ingoing
	flux.
	
	The event horizon is instead global.  If it can be written in the advanced
	patch as $r=r_{\mathrm{EH}}(v)$, its null generators obey
	\begin{equation}
		\frac{\dd r_{\mathrm{EH}}}{\dd v}
		=\frac12\left(1-\frac{2m(v)}{r_{\mathrm{EH}}(v)}\right).
		\label{eq:event-horizon-ode}
	\end{equation}
	The terminal or asymptotic condition selecting a particular solution depends
	on the complete future geometry.  Thus $r=2m(v)$ must not be called the event
	horizon without additional argument.
	
	There is no preferred Killing field in a general dynamical spherical
	space-time.  The Kodama field is the canonical replacement which uses only
	the spherical symmetry.  Write
	\begin{equation}
		g=h_{AB}(x)\,\dd x^A\dd x^B+r(x)^2\dd\Omega^2,
		\qquad A,B\in\{0,1\},
	\end{equation}
	and let $\epsilon_{AB}$ be the volume form of the two-dimensional orbit
	metric $h$.  The Kodama vector is
	\begin{equation}
		K^A=\epsilon^{AB}\nabla_Br,
		\qquad K^a=0\quad\hbox{in the angular directions}.
		\label{eq:Kodama-definition}
	\end{equation}
	It is orthogonal to $\nabla r$ and obeys
	\begin{equation}
		g(K,K)=-g^{-1}(\dd r,\dd r)=-C.
		\label{eq:Kodama-norm}
	\end{equation}
	Thus $K$ is timelike in the untrapped exterior, null at $C=0$, and spacelike
	in the trapped region.  With a corresponding choice of orbit-space
	orientation, $K=\partial_w$ in either Vaidya chart considered separately.
	These facts,
	together with $\nabla_aK^a=0$, explain why the Kodama flow supplies a
	geometrically preferred local time even when no time-translation isometry
	exists \cite{Kodama1980}.
	
	Hayward's surface gravity is the orbit-space scalar
	\begin{equation}
		\kappa_{\mathrm K}:=\frac12\Box_h r.
		\label{eq:Hayward-definition}
	\end{equation}
	This definition is adapted to trapping horizons rather than Killing horizons
	\cite{Hayward1998}.  Sign-sensitive Kodama identities on the horizon also
	depend on the choice of orbit-space orientation; we therefore use the scalar
	definition directly.  For
	\eqref{eq:unified-Vaidya}, $|\det h|=1$ and
	\begin{equation}
		\Box_h r
		=\partial_A(h^{AB}\partial_Br)
		=\partial_r C.
	\end{equation}
	Consequently a cross-section of \eqref{eq:trapping-horizon} has
	\begin{equation}
		\kappa_{\mathrm K}(w)
		=\frac12\left.\partial_r C(w,r)\right|_{r=2m(w)}
		=\frac{1}{4m(w)}.
		\label{eq:Kodama-kappa}
	\end{equation}
	The sign is positive precisely for an outer horizon.  This definition is
	local: in a dynamical geometry it need not equal the peeling function
	controlling null rays which eventually reach $\scrI^+$.
	
	Figure~\ref{fig:Vaidya-horizons-Kodama} collects the distinctions used below.
	In particular, the trapping horizon is locally determined, the event horizon
	is selected by the complete future geometry, and the Kodama field is nonzero
	when it becomes null.
	
	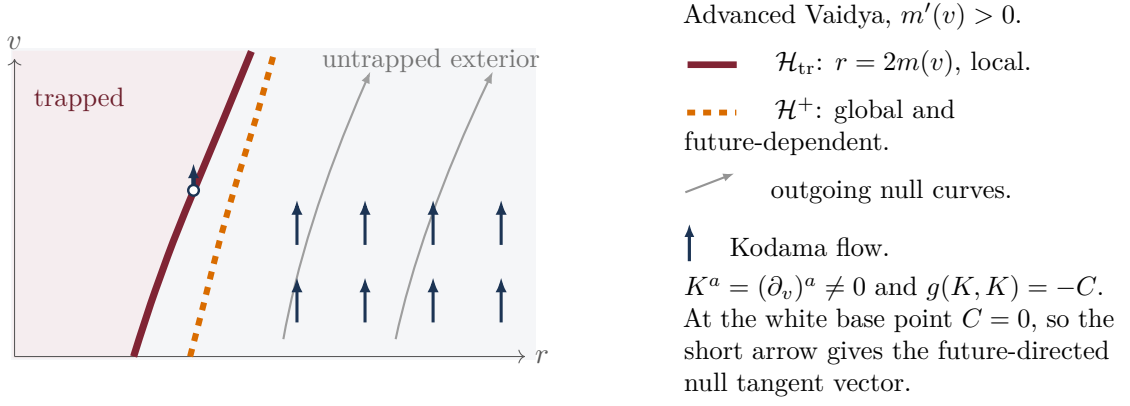
\begin{figure}[htbp]
		\centering
		\begin{minipage}[c]{0.60\textwidth}
			\centering
			\begin{tikzpicture}[x=0.90cm,y=0.82cm,
				kvec/.style={-{Latex[length=2mm]},deepblue,line width=1.15pt},
				nullray/.style={-{Latex[length=1.6mm]},black!38,line width=0.8pt}]
				\fill[softgray] (0.55,0.45) rectangle (8.25,5.47);
				\fill[burgundy!8]
				(0.57,0.47)--(2.35,0.47)--(2.60,1.35)--(2.93,2.35)--
				(3.30,3.35)--(3.68,4.35)--(4.08,5.43)--(0.57,5.43)--cycle;
				\draw[->,black!65] (0.60,0.50)--(8.10,0.50) node[right] {$r$};
				\draw[->,black!65] (0.60,0.50)--(0.60,5.30) node[above] {$v$};
				\node[text=burgundy!85!black,font=\small] at (1.55,4.65) {trapped};
				\node[text=black!55,font=\small] at (6.72,5.28) {untrapped exterior};
				
				\draw[burgundy,line width=2.7pt]
				plot[smooth] coordinates {(2.35,0.50) (2.60,1.35) (2.93,2.35)
					(3.30,3.35) (3.68,4.35) (4.08,5.42)};
				\draw[orange!85!black,line width=2.3pt,dashed]
				plot[smooth] coordinates {(3.18,0.50) (3.38,1.35) (3.61,2.35)
					(3.87,3.35) (4.15,4.35) (4.42,5.42)};
				
				\draw[nullray] (4.55,0.78) .. controls (4.73,1.90) and (5.12,3.33) .. (5.83,5.10);
				\draw[nullray] (6.20,0.78) .. controls (6.42,1.92) and (6.88,3.38) .. (7.63,5.10);
				
				\foreach \x in {4.75,5.75,6.75,7.75}{
					\draw[kvec] (\x,1.05)--(\x,1.78);
					\draw[kvec] (\x,2.30)--(\x,3.03);
				}
				\draw[kvec,line width=1.7pt] (3.23,3.18)--(3.23,3.62);
				\filldraw[fill=white,draw=deepblue,line width=.9pt] (3.23,3.18) circle (2pt);
			\end{tikzpicture}
		\end{minipage}\hfill
		\begin{minipage}[c]{0.36\textwidth}
			\raggedright
			\small
			Advanced Vaidya, $m'(v)>0$.
			
			\medskip
			\tikz[baseline=-.5ex]{\draw[burgundy,line width=2.5pt] (0,0)--(.65,0);}
			\quad $\mathcal H_{\mathrm{tr}}$: $r=2m(v)$, local.
			
			\medskip
			\tikz[baseline=-.5ex]{\draw[orange!85!black,line width=2.2pt,dashed] (0,0)--(.65,0);}
			\quad $\mathcal H^+$: global and future-dependent.
			
			\medskip
			\tikz[baseline=-.5ex]{\draw[-{Latex[length=1.5mm]},black!42,line width=.8pt] (0,0)--(.65,.25);}
			\quad outgoing null curves.
			
			\medskip
			\tikz[baseline=-.5ex]{\draw[-{Latex[length=1.8mm]},deepblue,line width=1.2pt] (.32,-.15)--(.32,.35);}
			\quad Kodama flow.
			
			\smallskip
			$K^a=(\partial_v)^a\neq0$ and $g(K,K)=-C$.  At the white base point
			$C=0$, so the short arrow gives the future-directed null tangent vector.
		\end{minipage}
		\caption{Trapping horizon, event horizon and Kodama flow in an accreting
			advanced-Vaidya patch.  The separation of the two horizon curves is
			geometric.}
		\label{fig:Vaidya-horizons-Kodama}
	\end{figure}
	
	\subsection{Admissible mass functions}
	
	We distinguish the following regimes.  The notation below classifies the mass
	profile only.  Proposition~\ref{prop:local-global-hyperbolicity} treats finite
	regular causal windows, whereas global causal conditions on an entire
	exterior are imposed separately in Assumption~\ref{ass:causal-domain}.
	
	\begin{definition}[Accreting class $\mathfrak V_{\mathrm A}$]
		The advanced mass function is smooth, positive after the onset of collapse,
		nondecreasing, and approaches a positive limiting mass:
		\begin{equation}
			m'(v)\geq0,
			\qquad
			\lim_{v\to+\infty}m(v)=M_+>0.
		\end{equation}
		If an early flat region is included, $m(v)=0$ there and the matching at the
		onset is assumed smooth to the differentiability order used in the field
		equation.
	\end{definition}
	
	\begin{definition}[Finite-slab evaporating class $\mathfrak V_{\mathrm E}$]
		On a prescribed compact interval $I=[w_-,w_+]$, the mass is smooth,
		nonincreasing and bounded away from zero:
		\begin{equation}
			m'(w)\leq0,
			\qquad
			0<m_I:=\inf_{w\in I}m(w)\leq \sup_{w\in I}m(w)<\infty.
		\end{equation}
		The advanced realisation describes negative ingoing energy near a future
		horizon; the retarded realisation describes positive outgoing null energy in
		an exterior radiation region.
	\end{definition}
	
	\begin{definition}[Asymptotically evaporating class
		$\mathfrak V_{\mathrm E_\infty}$]
		The mass function is smooth and satisfies
		\begin{equation}
			m(w)>0 \quad\text{for every finite }w,
			\qquad m'(w)\leq0,
			\qquad \lim_{w\to+\infty}m(w)=0.
			\label{eq:asymptotic-evaporation}
		\end{equation}
		For each compact $I\Subset\mathbb R$, the restriction belongs to
		$\mathfrak V_{\mathrm E}$.
	\end{definition}
	
	\begin{definition}[Schwarzschild--Vaidya--Schwarzschild class
		$\mathfrak V_{\mathrm S}$]
		\label{def:sandwich}
		A mass profile belongs to $\mathfrak V_{\mathrm S}$ when there are finite
		$w_-<w_+$ and constants $M_-,M_+>0$ such that
		\begin{equation}
			m(w)=M_-\quad(w\leq w_-),
			\qquad
			m(w)=M_+\quad(w\geq w_+),
			\qquad
			\inf_w m(w)>0.
			\label{eq:sandwich-profile}
		\end{equation}
		The intermediate Vaidya region may be increasing or decreasing.  A member of
		$\mathfrak V_{\mathrm S}$ is a complete scattering background only after the
		exterior patches and their time orientation have been specified.  The
		outgoing decreasing model treated in \cite{Coudray2024} is a white-hole
		exterior; its time reverse is an ingoing accreting black-hole exterior.  A
		single outgoing decreasing Vaidya patch must not, merely by changing its
		interpretation, be called an evaporating black hole.
	\end{definition}
	
	\begin{definition}[Evaporation--accretion turnaround class
		$\mathfrak V_{\mathrm T}$]
		\label{def:turnaround}
		A smooth positive mass profile belongs to $\mathfrak V_{\mathrm T}$ if there
		is a finite turnaround time $w_b$ such that
		\begin{equation}
			m'(w)\leq0\quad(w\leq w_b),\qquad
			m'(w_b)=0,\qquad
			m'(w)\geq0\quad(w\geq w_b),
			\label{eq:turnaround-signs}
		\end{equation}
		and $m_b:=m(w_b)>0$.  The asymptotically stationary subclass additionally
		satisfies $m(w)\to M_+\in(m_b,\infty)$ as $w\to+\infty$.
	\end{definition}
	
	The point $w_b$ is a local mass minimum, or turnaround, and need not be an
	inflection point (which would mean $m''(w_b)=0$).  In the advanced metric,
	the first phase of \eqref{eq:turnaround-signs} is an effective negative
	ingoing flux and the second is positive ingoing null dust.  This is a smooth
	ordinary Vaidya metric, although the null energy condition fails during the
	first phase.  A model containing both the compensating negative horizon flux
	and positive outgoing Hawking radiation, followed by independent incoming
	dust, requires matched ingoing/outgoing regions, as in
	\cite{Hiscock1981II}, or a more general metric with two independent
	null-flux components; it is then Vaidya-like rather than a single Vaidya
	solution.  Chirenti and Saa give a useful double-null formulation of the
	general one-stream Vaidya metric \cite{ChirentiSaa2012}, but that coordinate
	formulation must not itself be confused with a two-stream stress tensor.
	
	For $\kappa=1/(4m)$ define the dimensionless variation parameters
	\begin{equation}
		\delta_j(I)
		:=\sup_{w\in I}
		\frac{|\partial_w^j\kappa(w)|}{\kappa(w)^{j+1}},
		\qquad 1\leq j\leq N.
		\label{eq:adiabatic-hierarchy}
	\end{equation}
	In particular,
	\begin{equation}
		\delta_1(I)=4\sup_{w\in I}|m'(w)|.
		\label{eq:first-adiabatic-parameter}
	\end{equation}
	The quantities \eqref{eq:adiabatic-hierarchy} control variation in units of
	the local horizon scale.  They do not control curvature.  Indeed, for Vaidya
	the Kretschmann scalar has the Schwarzschild value
	\begin{equation}
		R_{abcd}R^{abcd}=\frac{48m(w)^2}{r^6},
	\end{equation}
	and hence
	\begin{equation}
		\left.R_{abcd}R^{abcd}\right|_{r=2m(w)}
		=\frac{3}{4m(w)^4}.
		\label{eq:horizon-curvature}
	\end{equation}
	Thus a profile in $\mathfrak V_{\mathrm E_\infty}$ may become increasingly
	adiabatic while its horizon curvature diverges as $w\to+\infty$.
	
	\subsection{Finite globally hyperbolic developments and global exteriors}
	
	We first distinguish two causal questions which are easily conflated.  If an
	asymptotically flat completion has past and future null infinities, its domain
	of outer communications is
	\begin{equation}
		\D_{\mathrm{oc}}
		:=I^-(\scrI^+)\cap I^+(\scrI^-),
		\label{eq:DOC-definition}
	\end{equation}
	where the chronological sets are evaluated in the conformal completion and
	then intersected with the physical space-time.  The function $m(w)$ does not
	by itself determine \eqref{eq:DOC-definition}: one must also prescribe the
	regular centre or singular boundary, the collapse and evaporation matchings,
	the asymptotic ends and the event horizon.  It is therefore impossible to
	deduce global hyperbolicity of an entire domain of outer communications from
	the sign of $m'(w)$ alone.  By contrast, the compact detector--horizon windows
	used in the finite-time propagation argument can be placed in globally
	hyperbolic domains of dependence under a directly checkable finite-hull
	condition.
	
	Let $O$ be a regular Vaidya coordinate region and let $S\subset O$ be an
	achronal set.  We write $D_O^+(S)$ for the points through which every
	past-inextendible causal curve in $O$ meets $S$, define $D_O^-(S)$ with the
	time orientation reversed, and put
	\begin{equation}
		D_O(S):=D_O^+(S)\cup D_O^-(S).
		\label{eq:domain-of-dependence}
	\end{equation}
	
	\begin{proposition}[Globally hyperbolic development of a finite causal window]
		\label{prop:local-global-hyperbolicity}
		Let $O\subset\{r>0\}$ be a Vaidya region, let
		$\tau_\epsilon$ be the temporal function
		\eqref{eq:temporal-function}, and let $K\Subset O$ be a compact finite causal
		window.  Choose $\tau_0<\inf_K\tau_\epsilon$ and an open subset
		$S\Subset\Sigma_0:=\{\tau_\epsilon=\tau_0\}$ with smooth boundary.  If
		\begin{equation}
			K\subset\operatorname{int}_O D_O^+(S),
			\label{eq:finite-hull-condition}
		\end{equation}
		then
		\begin{equation}
			\D_K:=\operatorname{int}_O D_O(S)
			\label{eq:finite-development}
		\end{equation}
		is an open, causally convex and globally hyperbolic subspace-time containing
		$K$.  The hypersurface $S\cap\D_K$ is a Cauchy hypersurface of $\D_K$.
		
		A sufficient way to verify \eqref{eq:finite-hull-condition} is the following.
		Every past-inextendible causal curve in $O$ through $K$ must meet $\Sigma_0$,
		and its footprint
		\begin{equation}
			A_K:=J_O^-(K)\cap\Sigma_0
			\label{eq:causal-footprint}
		\end{equation}
		must have compact closure contained in $S$.
	\end{proposition}
	
	\begin{proof}
		Since $\tau_\epsilon$ has timelike gradient, $\Sigma_0$ is spacelike and
		acausal.  The standard domain-of-dependence theorem then states that the
		interior of $D_O(S)$ is causally convex and globally hyperbolic and that the
		part of $S$ lying in this interior is a Cauchy hypersurface
		\cite[Ch.~1]{BaerGinouxPfaeffle2007}.  Hypothesis
		\eqref{eq:finite-hull-condition} puts $K$ in that interior.
		
		For the final assertion, every past-inextendible causal curve through a point
		of $K$ meets $\Sigma_0$ by hypothesis, and every such intersection belongs to
		$A_K\subset S$.  Hence $K\subset D_O^+(S)$.  Because
		$\overline{A_K}\subset S$, the intersections stay a positive distance from
		the edge of $S$.  The limit-curve theorem then gives the same property for
		all points in a neighbourhood of $K$; equivalently,
		$K\subset\operatorname{int}_O D_O^+(S)$.
	\end{proof}
	
	The causal mechanism may also be represented by extending the finite physical
	data segment $S_U$ to a Cauchy surface $\widehat\Sigma_U$ of an auxiliary
	globally hyperbolic development $\widehat{\M}_U$, as in
	Figure~\ref{fig:finite-window-completion}.  This auxiliary extension is not a
	completion of the full domain of outer communications and is not needed for
	the definition of $\D_K$.  It only makes finite propagation visible: once
	$K_U\subset D_{\widehat\M_U}(S_U)$, changing data on the dashed portions of
	$\widehat\Sigma_U$ cannot change the solution on $K_U$.
	
	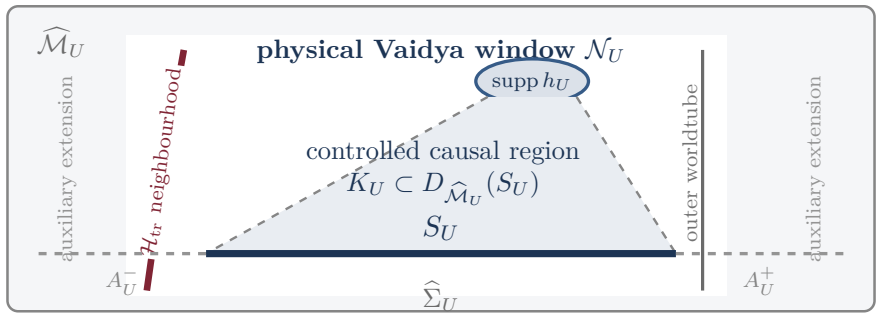
\begin{figure}[htbp]
		\centering
		\begin{tikzpicture}[x=0.98cm,y=0.74cm,
			cone/.style={black!48,dashed,line width=0.9pt},
			outer/.style={draw=black!48,line width=0.95pt,rounded corners=4pt}]
			\filldraw[outer,fill=softgray] (0.40,0.55) rectangle (12.10,6.00);
			
			\fill[white] (2.00,0.82) rectangle (10.10,5.48);
			\node[anchor=north west,text=black!58,font=\small] at (0.66,5.82)
			{$\widehat{\mathcal M}_U$};
			\node[text=deepblue,font=\bfseries\small] at (6.25,5.20)
			{physical Vaidya window $\mathcal N_U$};
			\node[rotate=90,text=black!42,font=\scriptsize] at (1.20,3.10)
			{auxiliary extension};
			\node[rotate=90,text=black!42,font=\scriptsize] at (11.30,3.10)
			{auxiliary extension};
			
			\draw[burgundy,line width=2.5pt]
			(2.28,0.92) .. controls (2.42,2.20) and (2.58,3.80) .. (2.78,5.22);
			\node[rotate=82,text=burgundy,font=\scriptsize,fill=white,inner sep=1pt]
			at (2.49,3.22) {$\mathcal H_{\mathrm{tr}}$ neighbourhood};
			\draw[black!58,line width=1.2pt] (9.78,0.92)--(9.78,5.22);
			\node[rotate=90,text=black!55,font=\scriptsize,fill=white,inner sep=1pt]
			at (9.61,3.20) {outer worldtube};
			
			\filldraw[fill=midblue!18,draw=midblue,line width=1.25pt]
			(7.48,4.66) ellipse (0.76 and 0.40);
			\node[text=deepblue,font=\scriptsize] at (7.48,4.66)
			{$\operatorname{supp}h_U$};
			\coordinate (detL) at (6.88,4.38);
			\coordinate (detR) at (8.08,4.38);
			\coordinate (sL) at (3.08,1.58);
			\coordinate (sR) at (9.42,1.58);
			\fill[midblue!11] (sL)--(sR)--(detR)--(detL)--cycle;
			\draw[cone] (detL)--(sL);
			\draw[cone] (detR)--(sR);
			\node[align=center,text=deepblue,font=\small] at (6.28,3.02)
			{controlled causal region\\$K_U\subset D_{\widehat{\mathcal M}_U}(S_U)$};
			
			\draw[black!42,dashed,line width=1.25pt] (0.82,1.58)--(3.08,1.58);
			\draw[deepblue,line width=2.7pt] (sL)--(sR);
			\draw[black!42,dashed,line width=1.25pt] (9.42,1.58)--(11.68,1.58);
			\node[above,text=deepblue] at (6.25,1.64) {$S_U$};
			\node[below,text=black!48,font=\scriptsize] at (1.95,1.52) {$A_U^-$};
			\node[below,text=black!48,font=\scriptsize] at (10.55,1.52) {$A_U^+$};
			\node[text=black!58,font=\small] at (6.25,0.82) {$\widehat\Sigma_U$};
		\end{tikzpicture}
		\caption{A finite detector--horizon window embedded in an auxiliary Cauchy
			development.  The solid segment contains the physical data relevant to
			$K_U$; the dashed segments complete the Cauchy surface but do not influence
			the solution on $K_U$.}
		\label{fig:finite-window-completion}
	\end{figure}
	
	The sufficient condition is geometric rather than circular.  On a compact
	regular cylinder
	\begin{equation}
		\tau_0\leq\tau_\epsilon\leq\tau_1,
		\qquad r_{\min}\leq r\leq r_{\max},
		\qquad r_{\min}>0,
		\label{eq:regular-causal-cylinder}
	\end{equation}
	the set of causal tangent vectors normalised by
	$\dd\tau_\epsilon(X)=1$ is compact.  Relative to any auxiliary Riemannian
	metric, their spatial speeds therefore have a uniform upper bound.  If the
	causal curves from $K$ remain in such a cylinder until they reach
	$\Sigma_0$, their footprint has bounded radial displacement; compactness of
	$\mathbb S^2$ then makes $\overline{A_K}$ compact.  Enlarging its footprint
	slightly inside $\Sigma_0$ produces the required $S$.  The condition that
	the curves remain in the regular cylinder is precisely what excludes a
	central singular boundary or an artificial patch boundary from intruding
	into the finite experiment.
	
	\begin{corollary}[Application to the five mass-profile classes]
		\label{cor:mass-class-hyperbolicity}
		Let $K$ be a finite detector--horizon window whose causal curves satisfy the
		regular-cylinder condition just described.  Then $K$ has a causally convex,
		globally hyperbolic development $\D_K$ for each of the classes
		$\mathfrak V_{\mathrm A}$, $\mathfrak V_{\mathrm E}$,
		$\mathfrak V_{\mathrm E_\infty}$, $\mathfrak V_{\mathrm S}$ and
		$\mathfrak V_{\mathrm T}$, provided the chosen compact coordinate interval
		does not contain a zero-mass singular endpoint.
	\end{corollary}
	
	\begin{proof}
		For $\mathfrak V_{\mathrm E}$, positivity on the compact interval gives
		$r_{\mathrm H}=2m(w)\geq2m_I>0$, so a compact horizon collar can be chosen
		with $r_{\min}>0$.  Every finite restriction of
		$\mathfrak V_{\mathrm E_\infty}$ belongs to this case.  For
		$\mathfrak V_{\mathrm A}$, any compact interval after the positive-mass onset
		has the same property; an earlier flat portion may also be included when its
		centre and matching are regular.  The definitions of
		$\mathfrak V_{\mathrm S}$ and $\mathfrak V_{\mathrm T}$ keep the mass
		strictly positive on every compact window under consideration, including a
		window across the turnaround.  Smoothness of $m$ makes all metric
		coefficients bounded on the compact cylinder in every case.  The speed-bound
		argument above verifies the compact-footprint part of the proposition, and
		the assumed absence of an intervening singular or patch boundary verifies
		that all relevant curves reach $\Sigma_0$.  Proposition
		\ref{prop:local-global-hyperbolicity} now supplies $\D_K$.
	\end{proof}
	
	Corollary~\ref{cor:mass-class-hyperbolicity} is the global-hyperbolicity
	statement actually needed for a single finite-window Cauchy evolution.  It
	does not turn a mass profile into a complete black-hole space-time.  The
	global geometry requires additional information in each mass class.
	
	For $\mathfrak V_{\mathrm A}$, a globally hyperbolic domain of outer
	communications requires a regular collapse matching and asymptotic
	predictability.  These properties do not follow from $m'(v)\geq0$.  For
	example, the self-similar profile $m(v)=\mu v$ admits locally outgoing null
	curves from $(v,r)=(0,0)$ when $1-16\mu\geq0$
	\cite{DwivediJoshi1989}.
	
	A member of $\mathfrak V_{\mathrm E}$ is a finite slab and consequently has
	no intrinsic future null infinity or domain of outer communications.  Its
	finite development is covered by the corollary, whereas a global statement
	depends on the extensions attached at its two ends.  For
	$\mathfrak V_{\mathrm E_\infty}$, the absence of a finite zero-mass endpoint
	removes the corresponding finite-endpoint obstruction, and every finite
	window is covered by the corollary.  Global hyperbolicity and affine
	completeness as $w\to+\infty$ require separate analysis because the curvature
	\eqref{eq:horizon-curvature} is not uniformly bounded as $m\to0$.
	
	The definition of $\mathfrak V_{\mathrm S}$ likewise does not determine the
	global causal structure.  The decreasing outgoing exterior considered in
	\cite{Coudray2024} is supplied with a global Cauchy and scattering
	construction and is Schwarzschild near both temporal ends.  That background
	and its time reverse provide the global sandwich examples used below.  An
	arbitrary interpolation satisfying \eqref{eq:sandwich-profile} does not
	inherit these properties from its mass profile alone.
	
	A $\mathfrak V_{\mathrm T}$ turnaround produces no local causal pathology
	while $m(w_b)>0$, and windows crossing $w_b$ are therefore covered by the
	corollary.  A complete black-hole exterior additionally requires a compatible
	global matching, event horizon and asymptotic end; the change of sign of
	$m'$ does not determine those structures.
	
	The field algebra for one finite experiment can therefore be constructed on
	$\D_K$.  Statements comparing an unbounded family of late-time detectors,
	or referring to scattering at a common $\scrI^+$, require one common global
	exterior.  We state that extra hypothesis explicitly.
	
	\begin{assumption}[Global causal domain]
		\label{ass:causal-domain}
		Whenever a result uses a common asymptotic state, global scattering or a
		late-time family of detector windows, there is a causally convex, globally
		hyperbolic open subspace-time $(\D,g)\subset(\M,g)$ containing:
		\begin{enumerate}[label=(\roman*)]
			\item the supports and finite causal developments of all detector smearings
			under consideration;
			\item the relevant outer-horizon neighbourhoods;
			\item an asymptotically flat end with a fixed normalisation of Bondi time.
		\end{enumerate}
		No causal curve enters $\D$ from a naked central singular boundary.  For a
		complete asymptotically flat black-hole model, $\D$ may be chosen as its
		globally hyperbolic domain of outer communications.  This assumption is a
		condition on the completed geometry, not on $m(w)$ alone.
	\end{assumption}
	
	Global hyperbolicity and geodesic completeness remain logically independent:
	a future-inextendible causal curve can terminate at a spacelike singularity
	in a globally hyperbolic space-time.  Finite-time complete evaporation is
	excluded from the main analysis for a different reason.  In the usual
	semiclassical diagram obtained by joining a decreasing Vaidya region to a
	regular post-evaporation region, causal diamonds can accumulate at the
	missing terminal point and fail to be compact.  This is not a theorem that
	every conceivable completion must fail: P\^egas et al. construct a different
	globally hyperbolic completion by modifying the final Planckian region
	\cite{PegasEtAl2025}.  Moreover, the law $m'=-\alpha/m^2$ reaches zero in
	finite time and loses both curvature and adiabatic control at its endpoint.
	
	\section{The Klein--Gordon field and Hadamard states}
	\label{sec:AQFT}
	
	Let $(\D,g)$ be either a finite development $\D_K$ supplied by
	Proposition~\ref{prop:local-global-hyperbolicity} or a common global domain
	satisfying Assumption~\ref{ass:causal-domain}, according to the result under
	consideration.  On this globally hyperbolic space-time, consider the normally
	hyperbolic operator
	\begin{equation}
		P=\Box_g-\mu^2-\xi R.
		\label{eq:KG-operator}
	\end{equation}
	Global hyperbolicity gives unique retarded and advanced Green operators
	\begin{equation}
		E^\pm:\Cinf(\D)\longrightarrow C^\infty(\D),
		\qquad
		PE^\pm=E^\pm P=\mathrm{id},
		\qquad
		\supp(E^\pm h)\subset J^\pm(\supp h).
		\label{eq:Green-operators}
	\end{equation}
	Their difference $E=E^- -E^+$ is the causal propagator.  The field algebra is
	generated by $\Phi(h)$, $h\in\Cinf(\D)$, subject to linearity, hermiticity,
	$\Phi(Ph)=0$ and
	\begin{equation}
		[\Phi(h_1),\Phi(h_2)]=iE(h_1,h_2)\mathbf 1.
	\end{equation}
	
	A quasifree state $\omega$ is fixed by its two-point distribution
	\begin{equation}
		W_2(h_1,h_2)=\omega\!\left(\Phi(h_1)\Phi(h_2)\right).
	\end{equation}
	We require $W_2$ to be Hadamard.  We recall why this is a precise local
	condition rather than an ansatz for a preferred vacuum.  A neighbourhood
	$\Ocal$ is geodesically convex when every two of its points are joined by a
	unique geodesic lying in $\Ocal$.  Synge's world function
	\begin{equation}
		\sigma(x,x')=\frac12\,\bigl(\hbox{signed squared geodesic distance}
		\bigr)
	\end{equation}
	is then smooth on $\Ocal\times\Ocal$.  Choose any smooth temporal function
	$T$ on $\Ocal$ and set
	\begin{equation}
		\sigma_\eps(x,x')
		=\sigma(x,x')+2i\eps\bigl(T(x)-T(x')\bigr)+\eps^2,
		\qquad \eps\downarrow0.
		\label{eq:iepsilon}
	\end{equation}
	Changing $T$ or the harmless details of this prescription changes the
	parametrix only by a smooth kernel.
	
	In four dimensions the local Hadamard condition says that the kernel has the
	form
	\begin{equation}
		W_2(x,x')=
		\frac{1}{8\pi^2}
		\left(
		\frac{U(x,x')}{\sigma_\eps(x,x')}
		+V(x,x')\log\frac{\sigma_\eps(x,x')}{\ell^2}
		\right)
		+H(x,x'),
		\label{eq:Hadamard-form}
	\end{equation}
	where the equality is understood as the distributional boundary value of
	\eqref{eq:iepsilon}.  The numerical coefficient assumes that $\sigma$ is one
	half the squared distance; authors who use the full squared distance write a
	different prefactor.  The displayed sign of $V$ is likewise tied to the
	conventions for $P$ and the $i\epsilon$ prescription; changing those
	conventions changes the transport coefficients without changing the
	Hadamard condition.
	
	The structure of \eqref{eq:Hadamard-form} is forced by the Klein--Gordon
	equation.  Substitution of the singular ansatz into $P_xW_2=0$ and matching
	powers of $\sigma$ gives transport equations along the unique geodesic.  The
	most singular coefficient obeys
	\begin{equation}
		2\sigma^{;a}\nabla_aU+(\Box\sigma-4)U=0,
		\qquad U(x,x)=1,
		\label{eq:U-transport}
	\end{equation}
	whose unique solution is $U=\Delta^{1/2}$, with $\Delta$ the van Vleck--Morette
	determinant.  The coefficient $V$ has a formal expansion
	$V=\sum_{j\geq0}V_j\sigma^j$; successive transport equations determine each
	$V_j$ from $P$, the metric and the preceding coefficient.  The residual
	freedom is smooth and is denoted by $H$.  Hence $U$ and $V$ are local geometric
	data, while the choice of state first appears in $H$.  In particular, the
	two-point functions of any two Hadamard states differ by a smooth bisolution.
	
	There is an equivalent coordinate-free formulation.  If $(x,k)\sim(x',k')$
	means that $x$ and $x'$ lie on one null geodesic and the covectors are related
	by parallel transport, then
	\begin{equation}
		\WF(W_2)=
		\left\{(x,k;x',-k'):(x,k)\sim(x',k'),\ k\in\overline V_+\right\}.
		\label{eq:microlocal-spectrum}
	\end{equation}
	Radzikowski's theorem proves that this microlocal spectrum condition is
	equivalent to the local representation \eqref{eq:Hadamard-form}, modulo a
	smooth kernel \cite{Radzikowski1996}.  It is this equivalence which justifies
	the phrase ``$W_2$ is Hadamard'' used below.  The condition fixes the
	ultraviolet singularity and permits local renormalisation; it does not fix
	the infrared content or the flux at infinity.
	
	Generalised Riesz distributions are the local building blocks for the
	Hadamard expansion of the Green operators.  They provide systematic control
	of the singular terms and their transport coefficients, but they do not give
	an explicit global solution of the time-dependent scattering problem.  The
	formal Riesz series need not converge; cutoffs and hyperbolic PDE estimates
	are needed to pass from a local parametrix to true local and global Green
	operators \cite{BaerGinouxPfaeffle2007}.  In the present problem their role is
	to isolate the universal singular kernel and, potentially, to calculate
	subleading curvature corrections.  They do not determine the smooth
	state-dependent term $H$.
	
	Hadamard regularity alone is sufficient for the universal local scaling limit
	but not for a unique flux at infinity.  We therefore use the quantitative
	continuity bound \eqref{eq:state-continuity} and record the second-channel
	contribution explicitly as \eqref{eq:outgoing-state-error}.  Its decay for an
	Unruh-type state depends on the detector switching and a further asymptotic
	state/scattering argument; it is not a consequence of Hadamard regularity.
	
	\section{Detector observables and backward propagation}
	\label{sec:detectors}
	
	Let $h_U\in\Cinf(\D)$ be a family of smearing functions supported near a
	distant detector at Bondi time $U$.  What Fredenhagen and Haag call a detector
	is not a new fundamental measuring apparatus added to the theory.  It is a
	local algebraic operation $Q_U$ whose occurrence is tested by the positive
	element $Q_U^*Q_U$.  For the free field take
	\begin{equation}
		Q_U=\Phi(h_U),
		\qquad
		\mathcal F_g[h_U]=\omega(Q_U^*Q_U)
		=W_2(\overline{h_U},h_U).
		\label{eq:response-functional}
	\end{equation}
	The abstract polynomial field is unbounded, so $Q_U^*Q_U$ is not literally a
	bounded POVM effect without an additional functional-analytic construction.
	What is needed here is the unambiguous nonnegative expectation
	\eqref{eq:response-functional}; a bounded detector can equivalently be made
	from Weyl operators.  The smearing $h_U$ incorporates the detector's spatial
	profile, switching and frequency filter.  A detector which is passive in the
	asymptotic vacuum is obtained by choosing that filter to have positive energy
	transfer with respect to the asymptotic time translation.
	
	This observable is closely related to the usual Unruh--DeWitt model.  For a
	two-level detector on a worldline $\gamma$, with gap $\Omega>0$ and switching
	$\chi$, second-order perturbation theory gives
	\begin{equation}
		\mathcal P_{g\to e}
		=\lambda^2|m_{eg}|^2
		\int\chi(\tau)\chi(\tau')e^{-i\Omega(\tau-\tau')}
		W_2(\gamma(\tau),\gamma(\tau'))\,\dd\tau\dd\tau'.
		\label{eq:UDW-response}
	\end{equation}
	After replacing the delta distribution on $\gamma$ by a smooth spatial
	profile, \eqref{eq:UDW-response} is precisely a quadratic smearing of $W_2$
	of the form \eqref{eq:response-functional}.  Thus the algebraic detector
	retains the two-point-function content of a conventional probe while avoiding
	irrelevant details of a particular internal detector Hamiltonian.
	
	\subsection{Why the observable may be propagated backwards}
	
	Let $E=E^--E^+$ and put $f_U=Eh_U$.  This is a smooth, spacelike-compact
	solution of $Pf_U=0$.  If $\Sigma$ is a Cauchy surface, Green's identity and
	$\Phi(Ph)=0$ give
	\begin{equation}
		\Phi(h_U)=\int_\Sigma
		\left(\Phi\nabla_af_U-f_U\nabla_a\Phi\right)n^a\,\dd\Sigma,
		\label{eq:detector-timeslice}
	\end{equation}
	up to the overall sign fixed by the convention for $E$.  To verify this,
	choose Cauchy surfaces $\Sigma_-\prec\supp h_U\prec\Sigma_+$, apply the
	divergence theorem to
	$j^a(\Phi,f)=\Phi\nabla^af-f\nabla^a\Phi$, and use
	$\nabla_aj^a=\Phi Pf-fP\Phi$.  The retarded and advanced support properties
	leave exactly the source integral $\Phi(h_U)$.  Conservation of $j^a$ in the
	source-free region then makes the right-hand side independent of $\Sigma$.
	This is the concrete form of the time-slice property: the same late detector
	observable can be represented by its classical Cauchy data at any earlier
	time.
	
	On a surface lying strictly to the past of $\supp h_U$, these data may equally
	be computed with the advanced Green solution $E^-h_U$ (with the present
	support convention).  In a stationary geometry, time translation relates all
	$f_U$ to one fixed solution.  In a dynamical geometry, $f_U$ is governed by a
	genuinely nonautonomous evolution.
	
	\subsection{Why a horizon/outgoing decomposition is expected}
	
	The decomposition used later has an elementary part and a difficult part
	which must not be conflated.  On an early Cauchy surface $\Sigma$, choose a
	smooth partition of unity
	\begin{equation}
		\chi_{\mathrm{hor}}+\chi_{\mathrm{mid}}+\chi_{\mathrm{out}}=1
		\label{eq:channel-partition}
	\end{equation}
	subordinate respectively to a horizon collar, a compact intermediate region,
	and the asymptotically flat end.  Multiplying both components of the Cauchy
	data by these cutoffs gives the exact identity
	\begin{equation}
		f_U=f_U^{\mathrm{hor}}+f_U^{\mathrm{out}}+r_U,
		\qquad r_U:=f_U^{\mathrm{mid}},
		\label{eq:dynamical-decomposition}
	\end{equation}
	at the level of Cauchy data.  Each summand has a unique homogeneous evolution,
	so the equality holds everywhere.  No approximation has yet been made.
	
	What stationary scattering proves, and what the Vaidya analysis has to replace,
	is that the middle datum becomes small as the detector is moved late, while
	the other two pieces converge after the correct rescalings to channel data.
	In Schwarzschild tortoise coordinate the two ends are
	$r^*\to-\infty$ (horizon) and $r^*\to+\infty$ (spatial/null infinity), and
	local energy decay turns \eqref{eq:channel-partition} into the familiar
	two-channel scattering decomposition
	\cite{DimockKay1987,DafermosRodnianski2013,Nicolas2016}.  For a time-dependent Vaidya
	potential, finite propagation speed still gives the exact partition, but
	smallness of $r_U$, control of cutoff commutators, and convergence of the
	channels require energy, redshift and local-decay estimates.
	Theorem~\ref{thm:target} supplies the finite-window energy and boundary-tail
	comparison.  The interior redshift-localisation part is retained separately
	as \eqref{eq:bridge-defect}; the distinction here explains why it cannot be
	deduced from channel scattering alone.
	
	After spherical harmonic decomposition, the Schwarzschild equation
	\begin{equation}
		\left(\partial_t^2-\partial_{r^*}^2+V_\ell(r)\right)f_{\ell m}=0
	\end{equation}
	is replaced by a $1+1$ dimensional hyperbolic equation with time-dependent
	coefficients.  There is no conserved Fourier frequency and no scalar
	transmission coefficient.  The dynamical greybody object is expected to be a
	kernel
	\begin{equation}
		\mathsf D_\ell(U;\omega,\omega'),
	\end{equation}
	which reduces in the stationary limit to
	\begin{equation}
		\mathsf D_\ell(U;\omega,\omega')
		\longrightarrow
		D_\ell(\omega)\,\delta(\omega-\omega').
	\end{equation}
	
	\section{The reduced Vaidya wave equation and a first energy estimate}
	\label{sec:wave-equation}
	
	The propagation problem can be written explicitly.  Let
	$Y_{\ell m}$ satisfy
	$\Delta_{\mathbb S^2}Y_{\ell m}=-\ell(\ell+1)Y_{\ell m}$ and set
	\begin{equation}
		\Phi(w,r,\nu)=\frac{1}{r}\psi_{\ell m}(w,r)Y_{\ell m}(\nu).
	\end{equation}
	
	\begin{proposition}[Characteristic radial equation]
		On every smooth Vaidya region with $r>0$, the equation $P\Phi=0$ is
		equivalent in the $(\ell,m)$ sector to
		\begin{equation}
			2\epsilon\,\partial_w\partial_r\psi_{\ell m}
			+\partial_r\!\left(C\partial_r\psi_{\ell m}\right)
			-\mathcal V_\ell(w,r)\psi_{\ell m}=0,
			\label{eq:reduced-Vaidya}
		\end{equation}
		where
		\begin{equation}
			\mathcal V_\ell(w,r)
			=\frac{\ell(\ell+1)}{r^2}+\frac{2m(w)}{r^3}+\mu^2+\xi R.
			\label{eq:Vaidya-potential}
		\end{equation}
		In the smooth null-fluid region $R=0$.
	\end{proposition}
	
	\begin{proof}
		The inverse orbit metric has
		$g^{ww}=0$, $g^{wr}=\epsilon$, $g^{rr}=C$, while
		$\sqrt{|g|}=r^2\sin\theta$.  Directly from
		$\Box\Phi=|g|^{-1/2}\partial_a(|g|^{1/2}g^{ab}\partial_b\Phi)$,
		\begin{equation}
			\Box\Phi=
			2\epsilon\Phi_{wr}+\frac{2\epsilon}{r}\Phi_w
			+C\Phi_{rr}+\left(\frac{2C}{r}+C_r\right)\Phi_r
			+\frac1{r^2}\Delta_{\mathbb S^2}\Phi.
			\label{eq:box-Vaidya}
		\end{equation}
		Insert $\Phi=r^{-1}\psi Y_{\ell m}$.  The terms proportional to
		$\psi_w/r^2$, $C\psi_r/r^2$, and $C\psi/r^3$ cancel pairwise.  Multiplication
		by $r$ leaves
		\begin{equation*}
			2\epsilon\psi_{wr}+C\psi_{rr}+C_r\psi_r
			-\left(\frac{C_r}{r}+\frac{\ell(\ell+1)}{r^2}
			+\mu^2+\xi R\right)\psi=0.
		\end{equation*}
		Since $C_r=2m(w)/r^2$, this is
		\eqref{eq:reduced-Vaidya}--\eqref{eq:Vaidya-potential}.
	\end{proof}
	
	Equation \eqref{eq:reduced-Vaidya} exhibits the precise obstruction to the
	Schwarzschild proof: the potential depends on $w$, so a temporal Fourier
	transform does not diagonalise the evolution.  The null slices $w=\mathrm
	{const.}$ are characteristic; energy estimates are instead made on the
	spacelike levels of $\tau_\epsilon$ from \eqref{eq:temporal-function}.
	
	For the minimally coupled equation, define
	\begin{equation}
		T_{ab}[\Phi]=\nabla_a\Phi\nabla_b\Phi
		-\frac12g_{ab}\left(\nabla^c\Phi\nabla_c\Phi+\mu^2\Phi^2\right).
	\end{equation}
	
	For any vector field \(X\), set
	\[
	J_X^a=T^a{}_b X^b
	\qquad\text{and}\qquad
	\pi^X_{ab}=\frac12\mathcal{L}_X g_{ab}.
	\]
	A direct differentiation yields the exact multiplier identity
	\begin{equation}
		\nabla_a J_X^a
		=T^{ab}\pi^X_{ab}+(\Box\Phi-\mu^2\Phi)X\Phi.
		\label{eq:multiplier-identity}
	\end{equation}
	Choose \(X\) to be uniformly future-timelike on a compact causal slab and integrate \eqref{eq:multiplier-identity} between two \(\tau_\epsilon\)-slices. For a solution the second term vanishes. The dominant energy condition makes the slice energy coercive, while bounded geometry supplies the estimate \(|T^{ab}\pi^X_{ab}|\leq C_X e_X[\Phi]\). With nonnegative lateral flux (or for data that remain compactly supported until that flux reaches the boundary) Gronwallâ€™s inequality then gives
	\begin{equation}
		E_X(\tau_2)
		\leq e^{C_X(\tau_2-\tau_1)}E_X(\tau_1).
		\label{eq:first-energy-estimate}
	\end{equation}
	This establishes continuous dependence and yields the basic propagator norm on every compact Vaidya slab. By itself, however, the estimate does not control the decay of the middle channel nor the exponentially squeezed bulk form of the horizon channel. The stationary boundary tail is treated below, while the bulk localisation contribution is encoded by the defect defined in \eqref{eq:bridge-defect}. When a smooth curvature coupling or mass potential is present the same compact-slab estimate continues to hold after the standard bounded lower-order potential term is added to the energy; in the Vaidya applications that follow one has \(R=0\), and the principal estimates are applied in the massless, minimally coupled sector.
	
	\begin{definition}[Compact-window energy spaces]
		\label{def:energy-spaces}
		Let \(D_{12}\) be the globally hyperbolic causal slab lying between two levels \(\Sigma_{t_1}\) and \(\Sigma_{t_2}\) of the function \(t=w-\epsilon r\). Fix a compact set \(K_1\Subset\Sigma_{t_1}\) and assume that the compact set \(J_D(K_1)\cap D_{12}\) remains inside the region \(r\geq r_{\min}>0\). For each integer \(s\geq0\) define
		\begin{equation}
			X^s(\Sigma_t)=H^{s+1}(\Sigma_t)\oplus H^s(\Sigma_t),
			\label{eq:concrete-Xs}
		\end{equation}
		where the two summands are the field and its future unit-normal derivative. The space is applied to data supported in \(K_t:=J_D(K_1)\cap\Sigma_t\). Sobolev norms are constructed with respect to the induced Riemannian metric; on any fixed compact family of slabs all such norms on the sets \(K_t\) are uniformly equivalent after transport along the foliation.
	\end{definition}
	
	\begin{proposition}[Uniform compact-window evolution]
		\label{prop:compact-evolution}
		Assume \(m\in C^{s+2}\) and \(0<m_0\leq m\leq m_1\) on the causal slab just defined. Then the Cauchy evolution of \(P\) for data supported in \(K_1\) defines a bounded map
		\begin{equation}
			\mathcal{U}(t_2,t_1):X^s(\Sigma_{t_1})\longrightarrow X^s(\Sigma_{t_2})
		\end{equation}
		whose image is supported in \(K_2=J_D(K_1)\cap\Sigma_{t_2}\). This map is an isomorphism onto its range. Fix a reference length \(\ell_0>0\) for this bounded coefficient class. If
		\begin{equation}
			\mathfrak{m}_s(t):=\sum_{j=1}^{s+2}\ell_0^{j-1}\sup_{K_t}|\partial_w^j m(w)|,
		\end{equation}
		then
		\begin{equation}
			\|\mathcal{U}(t_2,t_1)\|\leq C_s\exp\!\left(C_s|t_2-t_1|+C_s\int_{t_1}^{t_2}\mathfrak{m}_s(t)\,\mathrm{d}t\right).
			\label{eq:compact-evolution-bound}
		\end{equation}
		The reverse evolution satisfies the analogous estimate on the range. Here \(C_s\) depends only on uniform zeroth-order bounds for the metric, the foliation and the lower-order potential, and is independent of the particular data.
	\end{proposition}
	
	\begin{proof}
		Equation~\eqref{eq:temporal-calculation} ensures that every \(\Sigma_t\) is spacelike. On the stated compact coefficient class the lapse, the induced metric, the second fundamental form and all their derivatives through order \(s\) are uniformly bounded, while the lapse and induced metric remain uniformly nondegenerate. Apply \eqref{eq:multiplier-identity} to \(\Phi\) and to all coordinate derivatives through order \(s\). The differentiated equations involve only finitely many coefficient derivatives up to the displayed order and are controlled by \(C_s(1+\mathfrak{m}_s(t))E_s\). The boundary flux is nonnegative for an outgoing causal boundary. More invariantly, the standard Cauchy energy theorem on the globally hyperbolic slab applies directly to compactly supported data; finite propagation confines the solution to \(J_D(K_1)\), so no artificial timelike lateral boundary arises. Gronwallâ€™s inequality then yields \eqref{eq:compact-evolution-bound}. Reversing the foliation gives the corresponding estimate for the inverse. This is the standard energy argument for a normally hyperbolic operator, expressed here in the concrete function spaces used below; existence, uniqueness and finite propagation follow from \cite{BaerGinouxPfaeffle2007}.
	\end{proof}
	
	\section{Local horizon scaling}
	\label{sec:scaling}
	
	Fix a spherical cross-section $S_*$ of an outer trapping horizon and let
	$\kappa_*>0$ be its Kodama--Hayward surface gravity.  Choose adapted
	double-null coordinates $(\mathsf U,\mathsf V,\nu)$, with $S_*$ at
	$\mathsf U=\mathsf V=0$, and let $\mathcal T_*=\{\mathsf V=0\}$ be the null
	tangent hypersurface used in the scaling construction.
	
	\subsection{Linearisation of the projected Kodama flow}
	
	The actual Kodama field need not preserve $\mathcal T_*$.  Project it onto
	$T\mathcal T_*$ along the complementary null direction and call the result
	$\widehat K$.  Normalise the affine coordinate $\mathsf U$ so that
	$\mathsf U=0$ on $S_*$.  The horizon identity following
	\eqref{eq:Hayward-definition} fixes the first transverse derivative of the
	projected field:
	\begin{equation}
		\widehat K
		=\left(-\kappa_*\mathsf U+O(\mathsf U^2)\right)\partial_{\mathsf U}
		\quad\hbox{on }\mathcal T_*.
		\label{eq:projected-Kodama-generator}
	\end{equation}
	The sign corresponds to the orientation in which positive flow parameter
	moves points towards $S_*$.  A coordinate change
	$\mathsf U\mapsto a\mathsf U+O(\mathsf U^2)$ leaves the coefficient
	$\kappa_*$ unchanged.
	
	Let $\mathsf U_s$ be the integral curve satisfying
	$\dot{\mathsf U}_s=-\kappa_*\mathsf U_s+O(\mathsf U_s^2)$ and
	$\mathsf U_0=\mathsf U$.  Variation of constants gives, uniformly for $s$ in
	a compact interval,
	\begin{equation}
		\mathsf U_s=e^{-\kappa_*s}\mathsf U+O_s(\mathsf U^2).
		\label{eq:Kodama-dilation}
	\end{equation}
	Indeed, after multiplication by $e^{\kappa_*s}$ the nonlinear remainder is
	an integral bounded by a constant times $|\mathsf U|^2$.  This also explains
	why only the displayed linear term enters the scaling limit.  Put
	$\mathsf U=\lambda u$.  Then
	\begin{equation}
		\frac{\mathsf U_s(\lambda u)}{\lambda}
		=e^{-\kappa_*s}u+O_s(\lambda u^2)
		\longrightarrow e^{-\kappa_*s}u.
		\label{eq:linear-term-survives}
	\end{equation}
	Every quadratic or higher Taylor coefficient carries a positive power of
	$\lambda$ and disappears as $\lambda\downarrow0$.  Thus the limiting
	automorphism is exactly the dilation generated by
	$-\kappa_*\mathsf U\partial_{\mathsf U}$, even though the unscaled flow is not
	exactly linear.
	
	\subsection{Extraction of the universal kernel}
	
	Test the two-point function with derivatives
	$\partial_{\mathsf U}f$ and $\partial_{\mathsf U}f'$ of compactly supported
	functions.  This removes the zero mode and is the form naturally induced by
	the horizon symplectic structure.  Dilate their $\mathsf U$ support and
	concentrate a second test profile onto $\mathcal T_*$ in the $\mathsf V$
	direction, with the order and normalisation of limits fixed in
	\cite{KurpiczPinamontiVerch2021}.
	
	Near the diagonal of $S_*$ the world function has the Taylor form
	\begin{equation}
		\sigma(x,x')=
		-A_*\bigl(\mathsf U-\mathsf U'\bigr)
		\bigl(\mathsf V-\mathsf V'\bigr)
		+\frac{r_*^2}{2}d_{\mathbb S^2}(\nu,\nu')^2
		+O(|x-x'|^3),
		\label{eq:sigma-null-expansion}
	\end{equation}
	where $A_*>0$.  Insert \eqref{eq:sigma-null-expansion} into the first term of
	\eqref{eq:Hadamard-form}.  Transverse concentration, followed by integration
	of the two $\mathsf U$ derivatives by parts, produces the boundary value
	$(\mathsf U-\mathsf U'+i0)^{-2}$.  Angular concentration on the diagonal
	produces $r_*^2\delta(\nu,\nu')$.
	
	The remaining terms are suppressed for concrete scaling-degree reasons.
	Since $U(x,x')=\Delta^{1/2}(x,x')=1+O(|x-x'|^2)$, its nonconstant part gains a
	positive scaling power.  The $V\log\sigma_\eps$ term is less singular, the
	state-dependent $H$ is smooth, and the cubic remainder in
	\eqref{eq:sigma-null-expansion} again gains a positive scaling power.  A
	partition separating the diagonal from its complement and uniform
	distributional estimates justify passing to the limit; this is the rigorous
	analytic content of the cited scaling theorem.
	
	The resulting exact limit is
	\begin{equation}
		\Lambda_*(f,f')
		=-\frac{r_*^2}{\pi}
		\lim_{\eps\downarrow0}
		\int
		\frac{f(\mathsf U,0,\nu)f'(\mathsf U',0,\nu)}
		{(\mathsf U-\mathsf U'+i\eps)^2}
		\,\dd\mathsf U\,\dd\mathsf U'\,\dd\Omega(\nu),
		\label{eq:KPV-scaling}
	\end{equation}
	up to convention-dependent normalisation of the field and test functions.
	All state dependence lies in the smooth $H$, which has vanished; this is the
	source of universality.
	
	\subsection{Explicit KMS verification}
	
	Restrict first to $\mathsf U>0$ and set
	\begin{equation}
		\mathsf U=e^{\kappa_*u}.
		\label{eq:logarithmic-coordinate}
	\end{equation}
	The limiting dilation is the translation $u\mapsto u-s$.  Including the
	Jacobian factors gives
	\begin{equation}
		\frac{\dd\mathsf U\,\dd\mathsf U'}
		{(\mathsf U-\mathsf U'+i0)^2}
		=\frac{\kappa_*^2}{4}
		\frac{\dd u\,\dd u'}
		{\sinh^2\!\left(\frac{\kappa_*}{2}(u-u'+i0)\right)}.
		\label{eq:KPV-sinh-kernel}
	\end{equation}
	For $\mathsf U<0$, use $\mathsf U=-e^{\kappa_*u}$ and obtain the same
	formula.  Each connected side is therefore stationary under logarithmic
	translations.
	
	Set $\beta_*=2\pi/\kappa_*$ and consider
	\begin{equation}
		F(z)=\frac{\kappa_*^2}{4\sinh^2(\kappa_*z/2)}.
	\end{equation}
	It is holomorphic for $0<\operatorname{Im}z<\beta_*$, because its poles lie
	at $2\pi i n/\kappa_*$.  Also $F(z+i\beta_*)=F(z)$.  The lower and upper
	distributional boundary values are related by
	\begin{equation}
		\lim_{\eta\downarrow0}F(t+i\beta_*-i\eta)
		=F(t-i0)=F(-t+i0),
		\label{eq:KMS-boundaries}
	\end{equation}
	where the last equality uses evenness.  This is exactly the two-point KMS
	boundary condition: the upper boundary reverses the order of the two fields.
	
	Equivalently, Fourier transformation gives
	\begin{equation}
		\Lambda_*(f,f')
		=2r_*^2\int_{\mathbb R\times\mathbb S^2}
		\overline{\widehat f(E,\nu)}\widehat f'(E,\nu)
		\frac{E}{1-e^{-\beta_*E}}\,\dd E\dd\Omega(\nu),
		\label{eq:KMS-Fourier}
	\end{equation}
	with the same normalisation convention as \eqref{eq:KPV-scaling}.  Its
	spectral density satisfies detailed balance,
	\begin{equation}
		\rho(-E)=e^{-\beta_*E}\rho(E),
		\qquad \rho(E)=\frac{E}{1-e^{-\beta_*E}}.
	\end{equation}
	For the quasifree limiting state, the two-point KMS condition extends to the
	Weyl algebra.  The restriction to either side is therefore KMS at
	\begin{equation}
		\beta_*=\frac{2\pi}{\kappa_*}.
		\label{eq:local-temperature}
	\end{equation}
	
	Equation \eqref{eq:KPV-scaling} is local.  It does not assert that
	$\kappa_*$ equals the null-geodesic peeling function measured at $\scrI^+$.
	The quantity $b_U(L)$ in \eqref{eq:window-kappa-variation} measures only the
	difference between the frozen Kodama value and the local radial
	linearisation rate.  Relating either one to the asymptotic ray-tracing
	function requires the localisation bridge and far-zone propagation.
	
	\section{A rigorous reduction to a propagation estimate}
	\label{sec:reduction}
	
	We now record an abstract estimate which turns a norm bound for the propagated
	remainder into a two-sided response bound.  This elementary step is useful
	because it makes precise which part of the desired result follows only from
	positivity and which part requires new PDE analysis.
	
	Let $q$ be a positive sesquilinear form on a complex vector space $\mathscr X$.
	Write $\|x\|_q^2=q(x,x)$.
	
	\begin{lemma}[Positive-form comparison]
		\label{lem:positive-comparison}
		For $f=f_0+r\in\mathscr X$,
		\begin{equation}
			\left|q(f,f)-q(f_0,f_0)\right|
			\leq
			2\sqrt{q(f_0,f_0)q(r,r)}+q(r,r).
			\label{eq:positive-form-bound}
		\end{equation}
		If, in addition, $q(r,r)\leq C_q\|r\|_X^2$ for some norm $\|\cdot\|_X$, then
		\begin{equation}
			\left|q(f,f)-q(f_0,f_0)\right|
			\leq
			2\sqrt{C_q q(f_0,f_0)}\,\|r\|_X
			+C_q\|r\|_X^2.
			\label{eq:norm-response-bound}
		\end{equation}
	\end{lemma}
	
	\begin{proof}
		Sesquilinearity gives
		\begin{equation}
			q(f,f)-q(f_0,f_0)=q(f_0,r)+q(r,f_0)+q(r,r).
		\end{equation}
		The Cauchy--Schwarz inequality for positive sesquilinear forms gives
		\begin{equation*}
			|q(f_0,r)|\leq\sqrt{q(f_0,f_0)q(r,r)},
		\end{equation*}
		yielding
		\eqref{eq:positive-form-bound}.  The second assertion follows by inserting
		the assumed continuity estimate.
	\end{proof}
	
	Apply the lemma to the pullback of $W_2$ to Cauchy data on a suitable surface.
	Use the exact channel decomposition \eqref{eq:dynamical-decomposition}, and
	suppose, only for this preliminary reduction, that the vacuum part associated
	with $f_U^{\mathrm{out}}$ has been subtracted or shown not to excite a passive
	detector.  Taking
	$f_0=f_U^{\mathrm{hor}}$ and $r=r_U$, Lemma~\ref{lem:positive-comparison}
	immediately produces a certified response interval once $\|r_U\|_X$ is
	bounded.  Section~\ref{sec:certified-bound} removes this simplification: the
	outgoing contribution is retained there as a separately estimated positive
	term.
	
	This also displays why Hadamard regularity alone is not enough.  It controls
	the singular scaling of $q(f_U^{\mathrm{hor}},f_U^{\mathrm{hor}})$, while the
	constant $C_q$ depends on quantitative control of the state in the region
	sampled by the remainder.
	
	\section{Finite-window propagation and horizon localisation}
	\label{sec:target-propagation}
	
	Membership in $\mathfrak V_{\mathrm E}$ specifies a finite slab but not its
	future extension, while $\mathfrak V_{\mathrm E_\infty}$ provides no uniform
	positive lower bound for the mass.  Accordingly, the analysis uses two
	complementary results: an exact scattering theorem for compactly dynamical
	Vaidya backgrounds and a finite-window comparison theorem for evolving
	backgrounds.  Throughout Sections~\ref{sec:target-propagation}--
	\ref{sec:mass-profile-consequences}
	we take $\mu=0$.  Since $R_g=0$ in every smooth Vaidya region, the minimally
	and conformally coupled physical wave equations then coincide.
	
	\subsection{Exact boundary scattering for a Vaidya sandwich}
	
	For the massless conformal equation, a global boundary scattering result is
	available without a mode-by-mode construction.  Put
	$R=r^{-1}$, $\widehat g=R^2g$ and $\phi=r\Phi$.  Because $R_g=0$ in the smooth
	Vaidya region, $\Box_g\Phi=0$ is equivalent to the conformal wave equation for
	$\phi$ on $(\widehat\D,\widehat g)$.  The energy current used in
	\cite{Coudray2024} has the exact divergence
	\begin{equation}
		\widehat\nabla^aJ_a
		=m'(w)R^3|\phi_{\widehat n}|^2+m'(w)R|\phi|^2,
		\label{eq:Coudray-divergence}
	\end{equation}
	up to the sign fixed by time orientation; this is equation (45) of
	\cite{Coudray2024}.  On the foliation used there its positive flux is
	equivalent, on the compact Vaidya region, to an energy containing the
	schematic terms
	\begin{equation}
		|\phi_w|^2+R^4|\phi_R|^2+R^2|\nabla_{\mathbb S^2}\phi|^2
		+R^3|\phi|^2.
		\label{eq:conformal-energy-density}
	\end{equation}
	with coefficient weights depending on the chosen null frame.  Consequently,
	on the compact transition region,
	\begin{equation}
		|\widehat\nabla^aJ_a|\leq c_0|m'(w)|e[\phi].
		\label{eq:Coudray-error-bound}
	\end{equation}
	If $M(\tau)$ denotes the supremum of $|m'(u)|$ on the $\tau$-slice of
	Coudray's compact foliation, Stokes' theorem and Gronwall's inequality give
	\begin{equation}
		C_{\mathrm V}^{-1}E_{\Sigma}[\phi]
		\leq E_{\HH^+}[\phi]+E_{\scrI^+}[\phi]
		\leq C_{\mathrm V}E_{\Sigma}[\phi],
		\qquad
		C_{\mathrm V}=\exp\!\left(c_0\int_{\tau_-}^{\tau_+}M(\tau)\,\dd\tau\right),
		\label{eq:Vaidya-boundary-energy}
	\end{equation}
	where $c_0$ also depends on the fixed foliation, null frame and positive
	bounds for $m$.  The same estimate holds towards the past.  We do not identify
	this comparison constant with the total mass variation; the two are only
	comparable after the foliation is fixed.
	
	\begin{theorem}[Coudray's Vaidya scattering theorem]
		\label{thm:Coudray-scattering}
		For the outgoing, monotonically decreasing white-hole background with
		Schwarzschild past and future ends satisfying the hypotheses of
		\cite{Coudray2024}, the future and past trace maps
		\begin{equation}
			\mathcal T^\pm:\mathcal H_0\longrightarrow
			\mathcal H^\pm_{\HH}\oplus\mathcal H^\pm_{\scrI}
		\end{equation}
		are bounded isomorphisms of the corresponding energy spaces.  Hence
		\begin{equation}
			\mathcal S=\mathcal T^+(\mathcal T^-)^{-1}
			\label{eq:Vaidya-scattering-operator}
		\end{equation}
		is a bounded scattering isomorphism.  The assertion includes all angular
		momenta.
	\end{theorem}
	
	\begin{proof}
		Estimate \eqref{eq:Vaidya-boundary-energy} makes each trace map bounded,
		injective and of closed range.  Smooth compactly supported data prescribed on
		$\HH^\pm\cup\scrI^\pm$ determine a solution of the characteristic Goursat
		problem.  The Goursat existence theorem, finite propagation speed and
		\eqref{eq:Vaidya-boundary-energy} show that these smooth traces lie in the
		range and are dense in the boundary energy space.  A closed dense range is the
		whole space, so $\mathcal T^\pm$ is onto and has a bounded inverse.  Composition
		gives \eqref{eq:Vaidya-scattering-operator}.  This is the argument proved in
		detail in \cite{Coudray2024}; the essential Vaidya estimate is displayed in
		\eqref{eq:Coudray-divergence}--\eqref{eq:Vaidya-boundary-energy}.
	\end{proof}
	
	Thus an exact horizon/$\scrI$ channel decomposition is not conjectural for the
	specific Vaidya sandwich proved in \cite{Coudray2024}.  Time reversal gives
	the corresponding accreting black-hole statement.  The cited theorem does
	not cover an arbitrary sign-changing sandwich.  For a local, continuing
	evolution we now establish a
	quantitative comparison between the dynamical channel and a frozen
	Schwarzschild channel.
	
	\subsection{Coefficient and propagator comparison}
	
	Fix an observation time $U$ and a backwards \emph{spacelike} Cauchy slab
	$D_{U,L}$ bounded by levels $t_+(U)$ and $t_-(U)=t_+(U)-L$ of
	$t=w-\epsilon r$.  Let $K_t$ be the part of the $t$-slice met by its compact
	causal hull.  Choose the relevant horizon cross-section in the slab, denote
	its null coordinate by $w_+(U)$, and put
	$m_U=m(w_+(U))$ and $\kappa_U=(4m_U)^{-1}$.  We use the common
	$(w,r,\nu)$ chart to identify the Sobolev energy spaces of the dynamical and
	frozen metrics.  Let $X^s$ denote the resulting order-$s$ energy space.
	Fix once and for all a reference length $\ell_0$ on this bounded coefficient
	class; it only makes the coefficient seminorms dimensionally homogeneous and
	is absorbed into the constants.
	
	\begin{lemma}[Uniform coefficient difference]
		\label{lem:coefficient-difference}
		On a causal slab with $r\geq r_{\min}>0$ and
		$0<m_0\leq m(w)\leq m_1$, the first-order generators $A(t)$ of the Vaidya
		wave equation and $A_U$ of the frozen Schwarzschild equation, both written
		with respect to the spacelike $t$-foliation, satisfy
		\begin{equation}
			\|A(t)-A_U\|_{X^{s+1}\to X^s}
			\leq C_A d_{U,s}(t),
			\label{eq:generator-difference}
		\end{equation}
		where
		\begin{equation}
			d_{U,s}(t):=\sup_{x\in K_t}
			\left(|m(w(x))-m_U|+
			\sum_{j=1}^{s+2}\ell_0^j
			|\partial_w^jm(w(x))|\right),
		\end{equation}
		$C_A$ depends only on $s,r_{\min},m_0,m_1$ and the fixed foliation, not
		on the spherical harmonic number.
	\end{lemma}
	
	\begin{proof}
		From \eqref{eq:reduced-Vaidya},
		\begin{equation}
			C(w,r)-C_U(r)=-\frac{2(m(w)-m_U)}r,
			\qquad
			\mathcal V_\ell(w,r)-\mathcal V_{\ell,U}(r)
			=\frac{2(m(w)-m_U)}{r^3}.
			\label{eq:explicit-coefficient-difference}
		\end{equation}
		In the characteristic $(w,r)$ chart, all derivatives taken at fixed $w$ of
		the displayed differences are bounded by a constant times
		$|m(w)-m_U|$ when $r\geq r_{\min}$.  Passing to the spacelike coordinate
		$t=w-\epsilon r$ makes $w=t+\epsilon r$; spatially differentiating the
		coefficients therefore also produces derivatives of $m$.  Through Sobolev
		order $s$ these are bounded by $d_{U,s}(t)$.  Multiplication by the resulting
		coefficients and one spatial derivative map $X^{s+1}$ continuously to $X^s$.  The angular term
		$\ell(\ell+1)/r^2$ cancels in the difference.  Equivalently, before harmonic
		decomposition the rotations commute with both operators, so the same constant
		works after summing all angular modes.
	\end{proof}
	
	Define the slice-wise coefficient variations
	\begin{align}
		a_{U,s}(L)&:=\int_{t_-(U)}^{t_+(U)}d_{U,s}(t)\,\dd t,
		&v_U(L)&:=\int_{t_-(U)}^{t_+(U)}
		\sup_{x\in K_t}|m'(w(x))|\,\dd t,
		\label{eq:window-variations}\\
		\mathfrak v_{U,s}(L)&:=\int_{t_-(U)}^{t_+(U)}\mathfrak m_s(t)\,\dd t.
		\label{eq:window-high-variation}
	\end{align}
	If the two selected horizon cross-sections have advanced coordinates
	$v_-(U,L)<v_+(U,L)$, define separately
	\begin{equation}
		\Delta v_U(L):=v_+(U,L)-v_-(U,L),\qquad
		b_U(L):=\int_{v_-}^{v_+}
		|\kappa_{\mathrm{lin}}(v)-\kappa_U|\,\dd v.
		\label{eq:window-kappa-variation}
	\end{equation}
	The distinction is essential: $w$ is characteristic and is not the Cauchy
	evolution parameter.  For every fixed regular window all the quantities in
	\eqref{eq:window-variations}--\eqref{eq:window-kappa-variation} are finite:
	the causal hull is compact, $m\in C^{s+2}$ there, and the integration
	intervals are compact.  Uniform bounds for a family of windows require the
	same uniform lower bounds for $r$ and $m$, uniform upper bounds for the
	coefficient derivatives, and a uniform choice of foliation.
	
	\begin{proposition}[Vaidya--Schwarzschild propagator bound]
		\label{prop:Duhamel}
		Let $\mathcal U(t,s)$ and $\mathcal U_U(t,s)$ be the dynamical and frozen
		evolution families.  Define the finite stability factor
		\begin{equation}
			G_U(L):=\sup_{t\in[t_-(U),t_+(U)]}
			\|\mathcal U(t_-,t)\|_{X^s\to X^s}
			\|\mathcal U_U(t,t_+)\|_{X^{s+1}\to X^{s+1}}.
			\label{eq:stability-factor}
		\end{equation}
		On the above window,
		\begin{equation}
			\|\mathcal U(t_-,t_+)-\mathcal U_U(t_-,t_+)\|_{X^{s+1}\to X^s}
			\leq C_A G_U(L)a_{U,s}(L).
			\label{eq:Duhamel-bound}
		\end{equation}
		The compact-slab estimate gives the explicit, though generally non-decaying,
		upper bound $G_U(L)\leq C\exp(C L+C\mathfrak v_{U,s}(L))$: the dynamical
		factor acts on $X^s$, while the $X^{s+1}$ factor is the smooth frozen
		Schwarzschild evolution.
	\end{proposition}
	
	\begin{proof}
		Differentiate
		$\mathcal U(t_-,t)\mathcal U_U(t,t_+)$ and integrate from $t_-$ to $t_+$.
		This gives the exact Duhamel identity
		\begin{equation}
			\mathcal U(t_-,t_+)-\mathcal U_U(t_-,t_+)
			=\int_{t_-}^{t_+}\mathcal U(t_-,t)
			(A(t)-A_U)\mathcal U_U(t,t_+)\,\dd t,
		\end{equation}
		up to the harmless sign convention for generators.  Insert
		\eqref{eq:generator-difference}, take the two propagator suprema, and integrate
		to obtain
		\eqref{eq:Duhamel-bound}.
	\end{proof}
	
	In particular, $G_U(L)<\infty$ on every fixed regular slab.  No smallness of
	$G_U$ is assumed.  The perturbative quantity in
	\eqref{eq:Duhamel-bound} is the product $G_U(L)a_{U,s}(L)$; a late-time
	argument must control this product rather than either factor in isolation.
	
	\subsection{The quantitative dynamical propagation estimate}
	
	For each detector smearing $h_U$, define its final Cauchy data by
	\begin{equation}
		F_U:=\Gamma_{\Sigma_{t_+(U)}}E_gh_U.
		\label{eq:detector-final-data}
	\end{equation}
	The common chart and the equivalence of the dynamical and frozen energy norms
	on the compact final slice allow the same $F_U$ to be used as initial data for
	the frozen Schwarzschild evolution.  Thus the comparison below evolves
	identical data backwards with two different wave operators; it does not
	identify $E_gh_U$ with $E_{g_U}h_U$.
	
	The stationary input can then be made exact at the boundary-energy level.  Let
	\begin{equation}
		\mathcal R_M:\mathcal H_E(M)\longrightarrow
		L^2(\mathbb R_s\!\times\mathbb S^2)_{\HH}
		\oplus L^2(\mathbb R_s\!\times\mathbb S^2)_{\scrI}
		\label{eq:Sch-radiation-map}
	\end{equation}
	be the Schwarzschild trace map written in terms of the null derivatives of
	the conformally rescaled field, with the conventional factor $2M$ absorbed
	into the horizon norm.  Nicolas proves that this trace map is an isometry
	\emph{onto} the two boundary energy spaces by solving the Goursat problem
	\cite{Nicolas2016}.  Baskin and Wang independently prove the radiation-field
	energy identity, but explicitly do not characterise the range of their map;
	their result alone would not justify the inverse used below
	\cite{BaskinWang2014}.  Choose a bounded multiplication operator $\Pi_L$ on
	the two $L^2$ radiation variables, equal to one for $|s|\leq L$ and zero for
	$|s|\geq L+1$, with a fixed smooth transition.  Thus
	$1-\Pi_L$ is supported where $|s|\geq L$.  For
	frozen data $F_U$ put
	\begin{equation}
		\varepsilon^{\mathrm{Sch}}_{U}(L)
		:=\|(1-\Pi_L)\mathcal R_{m_U}F_U\|_{L^2\oplus L^2}.
		\label{eq:exact-Sch-tail}
	\end{equation}
	
	\begin{proposition}[Exact stationary channel tail]
		\label{prop:stationary-tail}
		The frozen backwards solution has the exact decomposition
		\begin{equation}
			f^{\mathrm{fr}}_{U,L}
			=f^{\mathrm{hor,fr}}_{U,L}+f^{\mathrm{out,fr}}_{U,L}
			+r^{\mathrm{fr}}_{U,L},
			\qquad
			\|r^{\mathrm{fr}}_{U,L}\|_{\mathcal H_E}
			=\varepsilon^{\mathrm{Sch}}_U(L).
			\label{eq:stationary-finite-remainder}
		\end{equation}
		If $j\geq1$ and the weighted radiation norm is finite, then
		\begin{equation}
			\varepsilon^{\mathrm{Sch}}_U(L)
			\leq \langle\kappa_UL\rangle^{-j}
			\|\langle\kappa_Us\rangle^j\mathcal R_{m_U}F_U\|_{L^2\oplus L^2}.
			\label{eq:Sch-moment-tail}
		\end{equation}
		The same statements hold after commuting with time translations and rotations,
		and hence in the corresponding order-$s$ radiation energy.
	\end{proposition}
	
	\begin{proof}
		Define the three terms by applying the inverse of
		\eqref{eq:Sch-radiation-map} to the horizon part of $\Pi_L\mathcal R F_U$,
		the null-infinity part, and $(1-\Pi_L)\mathcal R F_U$, respectively.
		Linearity gives the exact decomposition and the trace isometry gives the equality in
		\eqref{eq:stationary-finite-remainder}.  On the support of $1-\Pi_L$,
		$1\leq\langle\kappa_UL\rangle^{-j}
		\langle\kappa_Us\rangle^j$.  Multiplication and integration prove
		\eqref{eq:Sch-moment-tail}.  Commutation follows because the stationary wave
		operator commutes with Schwarzschild time translations and rotations.
	\end{proof}
	
	In Kruskal coordinate $\mathsf U=-e^{-\kappa_Us}$, translation of the
	horizon trace by the horizon-coordinate separation $\Delta v_U(L)$ is
	exactly dilation:
	\begin{equation}
		G_{\HH}(s-\Delta v_U(L),\nu)
		=G_{\HH}\!\left(-\kappa_U^{-1}
		\log(-\mathsf U/e^{-\kappa_U\Delta v_U(L)}),\nu\right).
		\label{eq:frozen-horizon-profile}
	\end{equation}
	This is an exact boundary statement.  It does not imply that
	$\mathcal R_M^{-1}(G_{\HH},0)$ is, on a finite interior Cauchy slice, a
	compactly supported function of the pointwise form
	$\psi(\rho/e^{-\kappa_U\Delta v_U(L)})$.  Establishing precisely that redshift
	localisation, with a norm estimate for the part outside a convex horizon
	neighbourhood, is the bridge required by the local Hadamard theorem.  It is
	isolated in Assumption~\ref{ass:localisation-bridge} below.
	
	As in the original Fredenhagen--Haag argument and in the KPV scaling theorem,
	we exclude the infrared zero mode.  Concretely, after expressing the
	localised horizon radiation variable in the Kruskal coordinate $\mathsf U$,
	we assume it is of the form $2\partial_{\mathsf U}p_U$ with
	$p_U\in C_0^\infty$; otherwise its nonzero integral must be retained as an
	additional infrared error.  Detector filters whose Fourier profile vanishes
	near zero frequency satisfy this condition.
	
	\begin{proposition}[Exact peeling coefficient and nonlinear remainder]
		\label{prop:peeling-remainder}
		In the advanced Vaidya metric, let $r_H$ be a positive solution of
		\eqref{eq:event-horizon-ode} on an interval $[v_-,v_+]$, put
		$\Delta v=v_+-v_-$ and $\rho=r-r_H$.  Then
		\begin{equation}
			\dot\rho=\kappa_{\mathrm{lin}}(v)\rho+N(v,\rho),
			\qquad
			\kappa_{\mathrm{lin}}(v)=\frac{m(v)}{r_H(v)^2},
			\qquad
			|N(v,\rho)|\leq\frac{m_1}{r_{\min}^3}|\rho|^2
			\label{eq:exact-peeling-equation}
		\end{equation}
		on every collar with $r_H+\rho\geq r_{\min}>0$.  The derivative at
		$\rho=0$ of the backwards ray map from $v_+$ to $v_-$ is
		\begin{equation}
			\lambda_{U,L}:=\exp\!\left[-\int_{v_-}^{v_+}
			\kappa_{\mathrm{lin}}(v)\,\dd v\right].
			\label{eq:dynamic-scaling-form}
		\end{equation}
		If $K=\sup_I|\kappa_{\mathrm{lin}}|$, the ray remains in the collar and
		$2(m_1/r_{\min}^3)\Delta v e^{K\Delta v}|\rho(v_+)|\leq1$, then
		\begin{equation}
			|\rho_{\mathrm{exact}}(v_-)-\lambda_{U,L}\rho(v_+)|
			\leq \frac{4m_1}{r_{\min}^3}\Delta v
			e^{3K\Delta v}|\rho(v_+)|^2.
			\label{eq:nonlinear-peeling-bound}
		\end{equation}
	\end{proposition}
	
	\begin{proof}
		Write $F(v,r)=\frac12(1-2m(v)/r)$.  Taylor's formula at $r_H$ gives
		$F(v,r_H+\rho)-F(v,r_H)=F_r(v,r_H)\rho+N$, where
		$F_r=m/r^2$ and $|F_{rr}|/2=m/r^3$.  This proves
		\eqref{eq:exact-peeling-equation}.  The variational equation at $\rho=0$ is
		$J'=\kappa_{\mathrm{lin}}J$; integration backwards gives
		\eqref{eq:dynamic-scaling-form}.  A bootstrap gives
		$|\rho(v)|\leq2e^{K\Delta v}|\rho(v_+)|$ under the stated smallness condition.
		Variation of constants, the quadratic bound for $N$, and
		$|J(v,s)|\leq e^{K|v-s|}$ then give
		\eqref{eq:nonlinear-peeling-bound}.
	\end{proof}
	For the retarded Vaidya form the relevant family of radial null curves and
	the orientation of the ray map are different; one must repeat this elementary
	linearisation with the appropriate null equation.  Proposition
	\ref{prop:peeling-remainder} is not a sign-independent formula for both
	Vaidya orientations.  Moreover $\kappa_{\mathrm{lin}}(v)$ is a local
	linearisation rate in the advanced coordinate.  It is not automatically the
	asymptotic ray-tracing function $-P''(U)/P'(U)$ defined below; relating them
	also requires propagation through the far zone.
	The logarithmic ratio satisfies
	\begin{equation}
		\left|\log\frac{\lambda_{U,L}}
		{e^{-\kappa_U\Delta v_U(L)}}\right|
		\leq b_U(L).
		\label{eq:lambda-ratio-bound}
	\end{equation}
	
	\begin{lemma}[Detector and stationary reconstruction constants]
		\label{lem:detector-reconstruction-constants}
		Let $\mathcal O_U\Subset D_{U,L}$ contain the support of the detector
		smearing.  The map which assigns final Cauchy data to the detector,
		\begin{equation}
			\mathcal T_{D,U}:H^{s+k}_0(\mathcal O_U)\longrightarrow X^{s+1},
			\qquad
			\mathcal T_{D,U}h
			:=\Gamma_{\Sigma_{t_+(U)}}E_gh,
			\label{eq:detector-data-map}
		\end{equation}
		is bounded.  Define
		\begin{equation}
			C_D(U):=\|\mathcal T_{D,U}\|_{H^{s+k}_0\to X^{s+1}}<\infty.
			\label{eq:detector-constant}
		\end{equation}
		Let $\mathcal J_{U,s}$ denote restriction of frozen Schwarzschild data to the
		initial finite slice, followed by the fixed identification with $X^s$.  On
		the order-$s$ commuted radiation space, the reconstruction map
		\begin{equation}
			\mathcal K_{R,U}
			:=\mathcal J_{U,s}\mathcal R_{m_U}^{-1}
			\label{eq:reconstruction-map}
		\end{equation}
		is bounded, and we set
		\begin{equation}
			C_R(U):=\|\mathcal K_{R,U}\|_{\mathrm{rad},s\to X^s}<\infty.
			\label{eq:reconstruction-constant}
		\end{equation}
		For detector regions, masses and finite-slice geometries in fixed bounded
		families, with supports contained in one compact causal tube, the suprema of
		$C_D(U)$ and $C_R(U)$ are finite.  We denote these uniform suprema by $C_D$
		and $C_R$.  When only one window is considered, the same symbols below may
		denote the corresponding fixed-window norms $C_D(U)$ and $C_R(U)$.
	\end{lemma}
	
	\begin{proof}
		The inhomogeneous energy estimate for $\Box_g u=h$, followed by the trace
		theorem on $\Sigma_{t_+(U)}$, gives
		$\|\Gamma_{\Sigma_{t_+}}E_gh\|_{X^{s+1}}
		\leq C_D(U)\|h\|_{H^{s+k}}$; finite propagation confines the estimate to the
		compact causal hull of $\mathcal O_U$.  Smooth coordinate changes and norm
		equivalence on that hull make the constant finite.  For reconstruction,
		Proposition~\ref{prop:stationary-tail} uses the onto Schwarzschild trace map
		$\mathcal R_{m_U}$, whose inverse is an isometry in the canonical radiation
		energy.  Restriction to a fixed finite slice and identification of the
		commuted canonical energy with $X^s$ are bounded, giving
		\eqref{eq:reconstruction-constant}.  Standard uniform energy and trace
		estimates give the final assertion on a uniformly bounded geometric family.
	\end{proof}
	
	\begin{theorem}[Quantitative finite-window propagation comparison]
		\label{thm:target}
		Under Assumption~\ref{ass:causal-domain}, suppose the causal window obeys the bounds of
		Lemma~\ref{lem:coefficient-difference}.  Let $h_U$ range over a
		bounded subset of $H^{s+k}_0$ in the detector region whose frozen radiation
		fields have the weighted norm in \eqref{eq:Sch-moment-tail}.  Then its exact
		backwards Cauchy data on the initial end of the window decompose as
		\begin{equation}
			f_U=f^{\mathrm{hor,fr}}_{U,L}+f^{\mathrm{out,fr}}_{U,L}+r_{U,L},
			\label{eq:proved-dynamical-decomposition}
		\end{equation}
		and
		\begin{equation}
			\|r_{U,L}\|_{X^s}
			\leq \mathcal B_U^{\mathrm{PDE}}(L;h_U),
			\label{eq:target-propagation-bound}
		\end{equation}
		where
		\begin{align}
			\mathcal B_U^{\mathrm{PDE}}(L;h_U)={}&
			C_D C_A G_U(L)a_{U,s}(L)\|h_U\|_{H^{s+k}}
			+C_R\varepsilon^{\mathrm{Sch}}_U(L).
			\label{eq:explicit-propagation-error}
		\end{align}
		The constants $C_D$ and $C_R$ have the fixed-window or uniform meaning
		specified in Lemma~\ref{lem:detector-reconstruction-constants}.  Uniformity
		holds when the window geometry, detector family and mass remain in the
		bounded families specified there.
	\end{theorem}
	
	\begin{proof}
		Use the frozen decomposition \eqref{eq:stationary-finite-remainder}.
		Proposition~\ref{prop:Duhamel} bounds the difference between the exact and
		frozen Cauchy data by the first term in
		\eqref{eq:explicit-propagation-error}, using
		\eqref{eq:detector-constant}.  Add the stationary remainder
		from \eqref{eq:stationary-finite-remainder}.  Defining $r_{U,L}$ as the sum of
		these two differences and applying \eqref{eq:reconstruction-constant} gives
		the exact identity and the stated bound.
	\end{proof}
	
	The theorem supplies a hard finite-window PDE bound: its first term measures frequency
	mixing caused by the time-dependent wave operator and its second term is the
	quantitatively decaying Schwarzschild boundary tail.  It deliberately makes
	no claim that the horizon trace has already become a local squeezed bulk test
	function.  Null peeling and that localisation enter next.
	
	\section{Certified detector bound}
	\label{sec:certified-bound}
	\label{sec:conditional-result}
	
	Let $q_U$ be the pullback of $W_2$ to the Cauchy data on
	$\Sigma_{t_-(U)}$.
	
	\begin{proposition}[Automatic compact-support state continuity]
		\label{prop:state-continuity}
		For every compact Cauchy-data region $K$ there are an integer $s_\omega$ and
		$C_{\omega,K}<\infty$ such that, for data supported in $K$,
		\begin{equation}
			q_U(r,r)\leq C_{\omega,K}\|r\|_{X^s}^2.
			\label{eq:state-continuity}
		\end{equation}
		One may choose $C_{\omega,K}$ uniformly for a family of Hadamard two-point
		functions bounded in the corresponding distribution seminorms and for a
		compact family of slices.
	\end{proposition}
	
	\begin{proof}
		The Hadamard wavefront set has only null covectors.  Its pullback to a
		spacelike Cauchy surface, including either normal derivative, is therefore
		well-defined by the distributional pullback theorem.  Every distribution has
		finite order on a compact set, so each of the four Cauchy-data kernels is
		bounded by finitely many $C^N$ seminorms of the two entries.  Sobolev
		embedding, with $s>N+3/2$ on the three-dimensional Cauchy surface, bounds
		those seminorms by $H^s$ norms.  Summing the four terms and increasing $s$ by
		one for the field component gives
		\eqref{eq:state-continuity}.  Boundedness of the relevant distribution
		seminorms gives the uniform assertion.
	\end{proof}
	
	\begin{assumption}[State control on the error subspace]
		\label{ass:state-energy-control}
		The Cauchy-data form obeys
		\begin{equation}
			q_U(z,z)\leq C_{\omega,U}\|z\|_{X^s}^2
			\label{eq:global-state-energy-control}
		\end{equation}
		on the linear span of the Duhamel error, the bridge error and the inverse
		radiation tail used below, together with the frozen outgoing channel
		$f^{\mathrm{out,fr}}_{U,L}$.  This is automatic from Proposition
		\ref{prop:state-continuity} when that span has support in one fixed compact
		set.  For the generally noncompact inverse radiation tail and outgoing
		channel it is an additional infrared/energy-continuity hypothesis on the
		state.  Without it one must retain their exact state quadratic forms rather
		than estimate them by radiation-energy norms.
	\end{assumption}
	
	Let $\Gamma_\Sigma g$ denote the Cauchy data on $\Sigma$ of the solution
	$Eg$ generated by a compact spacetime test function $g$.  Starting from the
	localised frozen horizon radiation profile, choose a smooth
	angular/transverse cut-off and form the KPV test function
	\begin{equation}
		g^{\mathrm{loc}}_{U,L,\mu}
		=v_\mu u_{\lambda_{U,L}}
		(2\partial_{\mathsf U}p_U),
		\qquad
		\lambda_{U,L}=\exp\!\left(-\int_{v_-(U,L)}^{v_+(U,L)}
		\kappa_{\mathrm{lin}}(v)\,\dd v\right),
		\label{eq:local-bridge-test}
	\end{equation}
	in a fixed convex horizon neighbourhood.  Here $u_\lambda$ and $v_\mu$ are
	exactly the longitudinal scaling and transverse concentration maps of
	Theorem~4.1 of \cite{KurpiczPinamontiVerch2021}; in particular that theorem
	has the iterated limit $\lambda\downarrow0$ followed by $\mu\downarrow0$.
	Equation~\eqref{eq:local-bridge-test} uses the advanced-orientation peeling
	formula.  In a retarded model $\lambda_{U,L}$ must instead be computed from
	its appropriately oriented ray-tracing map.
	
	\begin{definition}[Redshift-localisation bridge defect]
		\label{ass:localisation-bridge}
		On the initial slice $\Sigma_{t_-(U)}$ define the exact, nonnegative number
		\begin{equation}
			\eta_U(L,\mu):=
			\bigl\|f^{\mathrm{hor,fr}}_{U,L}
			-\Gamma_{\Sigma_{t_-(U)}}g^{\mathrm{loc}}_{U,L,\mu}\bigr\|_{X^s}.
			\label{eq:bridge-defect}
		\end{equation}
		We say that a detector family satisfies the quantitative localisation bridge
		if one can choose $L(U)\to\infty$ and $\mu(U)\downarrow0$ so that
		$\eta_U(L(U),\mu(U))\to0$.  A finite-window estimate may instead retain
		$\eta_U(L,\mu)$ explicitly without assuming it is small.
	\end{definition}
	
	To separate pure localisation from the peeling mismatch, let
	$g^{\mathrm{loc,fr}}_{U,L,\mu}$ be \eqref{eq:local-bridge-test} with
	$\lambda_{U,L}$ replaced by
	$e^{-\kappa_U\Delta v_U(L)}$ and put
	\begin{equation}
		\eta_U^{\mathrm{fr}}(L,\mu)=
		\|f^{\mathrm{hor,fr}}_{U,L}
		-\Gamma_{\Sigma_{t_-(U)}}g^{\mathrm{loc,fr}}_{U,L,\mu}\|_{X^s}.
	\end{equation}
	To make the comparison constant explicit, write $\theta=\log\lambda$ and
	let $B_{U,\mu}(\theta)$ be the linear map from $h_U$ to
	$\Gamma_{\Sigma_{t_-(U)}}v_\mu u_{e^\theta}
	(2\partial_{\mathsf U}p_U)$.  Put
	\begin{equation}
		I_{U,L}:=
		\left[\min\!\left\{\log\lambda_{U,L},-\kappa_U\Delta v_U(L)\right\},
		\max\!\left\{\log\lambda_{U,L},-\kappa_U\Delta v_U(L)\right\}\right]
	\end{equation}
	and define the fixed-window operator norm
	\begin{equation}
		C_{\Gamma,F;U}(L,\mu):=
		\sup_{\theta\in I_{U,L}}
		\|\partial_\theta B_{U,\mu}(\theta)\|_{H^{s+k}\to X^s}.
		\label{eq:GammaF-constant}
	\end{equation}
	This number is finite for every fixed regular window and $\mu>0$: dilation
	depends smoothly on $\theta$ on the compact interval $I_{U,L}$, and
	$\Gamma_\Sigma$ is bounded into Cauchy data.  For a bounded family of
	profiles and uniformly bounded trace geometries one may also define the
	uniform constant
	\begin{equation}
		C_{\Gamma,F}(\mu):=
		\sup_U C_{\Gamma,F;U}(L(U),\mu),
	\end{equation}
	provided the displayed supremum is finite.  Neither the fixed-window nor the
	uniform constant is asserted to remain bounded as $\mu\downarrow0$ without
	an additional estimate on the scaled profile family.  The mean-value theorem
	and \eqref{eq:lambda-ratio-bound} give the sharper fixed-window bound
	\begin{equation}
		\eta_U(L,\mu)
		\leq \eta_U^{\mathrm{fr}}(L,\mu)
		+C_{\Gamma,F;U}(L,\mu)b_U(L)
		\|h_U\|_{H^{s+k}}.
		\label{eq:bridge-peeling-bound}
	\end{equation}
	Thus $b_U$ is a controlled contribution to the bridge error, while the decay
	of $\eta_U^{\mathrm{fr}}$ is a separate localisation condition.
	
	The global isomorphism property of the trace map does not provide this
	condition because it does not quantify concentration of the inverse map in a
	fixed bulk neighbourhood.  Decay of \eqref{eq:bridge-defect} follows from a
	quantitative redshift or local-energy estimate compatible with the chosen
	channel cut-off.  Neither \cite{Nicolas2016},
	\cite{BaskinWang2014}, nor \cite{KurpiczPinamontiVerch2021} states that
	estimate, so the detector bound retains $\eta_U$ as an exact term.
	
	Let
	\begin{equation}
		R_U(L)=
		C_D C_A G_U(L)a_{U,s}(L)\|h_U\|_{H^{s+k}}
		+C_R\varepsilon^{\mathrm{Sch}}_U(L).
		\label{eq:response-propagation-error}
	\end{equation}
	This is the norm distance between the exact data and the sum of the
	\emph{frozen} horizon and outgoing channels.  Define the measurable
	asymptotic-state error
	\begin{equation}
		\mathcal V_U(L):=
		q_U(f^{\mathrm{out,fr}}_{U,L},f^{\mathrm{out,fr}}_{U,L}).
		\label{eq:outgoing-state-error}
	\end{equation}
	Whether this term tends to zero is a separate state-and-detector assertion.
	Compactly switched detectors generally retain vacuum fluctuations, so
	ground-state spectral support alone does not make it vanish.  For a different
	state it also records incoming or coherent radiation.
	Define the exact local thermal reference form
	\begin{equation}
		\mathcal F_{\mathrm{th},U}[p_U]
		:=\Lambda_*(\overline{p_U},p_U).
		\label{eq:thermal-reference-definition}
	\end{equation}
	If $p_U$ is obtained from the frozen Schwarzschild horizon channel with the
	Fredenhagen--Haag normalization, then the stationary calculation following
	\eqref{eq:FH-final} identifies this with
	$\mathcal F_{\mathrm{FH},m_U}[h_U]$.  Keeping
	\eqref{eq:thermal-reference-definition} until that normalization is imposed
	keeps the normalization convention distinct from the detector estimate.
	
	\begin{proposition}[Future-Unruh state on an eventually stationary model]
		\label{prop:future-Unruh}
		Suppose a globally hyperbolic Vaidya or matched background is exactly
		Schwarzschild of mass $M_+$ in a neighbourhood of, and to the future of, a
		Cauchy surface $\Sigma_+$, and that this stationary Cauchy neighbourhood is
		isometric to the corresponding Schwarzschild domain on which the Unruh
		covariance is defined.
		There exists a quasifree Hadamard state $\omega_{\mathrm{fU}}$ on the full
		field algebra whose Cauchy covariance on $\Sigma_+$ is that of the
		Schwarzschild Unruh state.
	\end{proposition}
	
	\begin{proof}
		Dappiaggi, Moretti and Pinamonti construct the Schwarzschild Unruh state as a
		positive quasifree Hadamard state \cite{DappiaggiMorettiPinamonti2011}.
		Use the stated isometry to restrict its two-point distribution and normal
		derivatives to $\Sigma_+\times\Sigma_+$, obtaining Cauchy covariances
		with the correct antisymmetric part.  The global Cauchy evolution is a real
		symplectic isomorphism; pullback of those covariances therefore defines a
		positive quasifree state on the full Weyl algebra and preserves the CCR.
		Propagation of the Hadamard wavefront set for a normally hyperbolic operator
		proves that the resulting state is Hadamard everywhere.  This construction
		does not by itself prove $\mathcal V_U(L)\to0$ for a compactly switched
		detector; that requires the precise detector spectral limit and is retained
		as a hypothesis in the late-time statements.
	\end{proof}
	
	\begin{proposition}[Finite Hadamard scaling defect]
		\label{lem:finite-Hadamard}
		For the actual Hadamard two-point function and the local test
		\eqref{eq:local-bridge-test}, define the exact finite-scale defect
		\begin{equation}
			\mathfrak h_U(L,\mu):=
			\left|W_2(\overline{g^{\mathrm{loc}}_{U,L,\mu}},
			g^{\mathrm{loc}}_{U,L,\mu})
			-\mathcal F_{\mathrm{th},U}[p_U]\right|.
			\label{eq:Hadamard-modulus-definition}
		\end{equation}
		It is finite for every positive $\lambda_{U,L}$ and $\mu$.  For every fixed
		Hadamard kernel and fixed test profile, Theorem~4.1 of
		\cite{KurpiczPinamontiVerch2021} gives the iterated limit
		\begin{equation}
			\lim_{\mu\downarrow0}\lim_{\lambda_{U,L}\downarrow0}
			\mathfrak h_U(L,\mu)=0.
			\label{eq:finite-Hadamard-error}
		\end{equation}
		With the stationary channel normalization just described,
		$\mathcal F_{\mathrm{th},U}[p_U]$ is the frozen Fredenhagen--Haag quadratic
		form.
	\end{proposition}
	
	\begin{proof}
		Finiteness follows because the scaled functions remain smooth and compactly
		supported for positive parameters and $W_2$ is a distribution.  The iterated
		limit is precisely parts (I) and (II) of the cited theorem.  Its limiting
		kernel is \eqref{eq:KPV-scaling}; inserting the horizon profile obtained from
		the frozen trace gives the same Fourier quadratic form as
		\eqref{eq:FH-final}.  The cited theorem is pointwise in the Hadamard kernel
		and test functions.  Uniform convergence for a $U$-dependent family is not
		automatic and must be assumed or proved by uniform bounds on the Hadamard
		coefficients and the scaled test family.
	\end{proof}
	
	For positive $L$ and $\mu$ define
	\begin{equation}
		\mathcal F_U^\sharp=
		\mathcal F_{\mathrm{th},U}[p_U]+\mathfrak h_U(L,\mu),
		\qquad
		S_U(L,\mu)=R_U(L)+\eta_U(L,\mu).
		\label{eq:Fsharp-S-definition}
	\end{equation}
	Set
	\begin{align}
		\mathcal E_U(L,\mu)
		:={}&2\sqrt{C_{\omega,U}}
		\left(\sqrt{\mathcal F_U^\sharp}+\sqrt{\mathcal V_U(L)}\right)S_U(L,\mu)
		+C_{\omega,U}S_U(L,\mu)^2\notag\\
		&+\mathfrak h_U(L,\mu)
		+2\sqrt{\mathcal F_U^\sharp\mathcal V_U(L)}+\mathcal V_U(L).
		\label{eq:certified-error}
	\end{align}
	
	\begin{lemma}[Finiteness of the certified-error data]
		\label{lem:certified-error-finiteness}
		Fix a regular detector--horizon window, $L<\infty$ and $\mu>0$, and impose
		Assumption~\ref{ass:state-energy-control}.  Then every quantity entering the
		certified response error is finite.  More explicitly,
		\begin{align}
			&C_A,\ C_D,\ C_R,\ G_U(L),\ a_{U,s}(L),\
			\varepsilon_U^{\mathrm{Sch}}(L),\ b_U(L),\ C_{\Gamma,F;U}(L,\mu),\
			C_{\omega,U}<\infty,
			\label{eq:finite-structural-constants}\\
			&R_U(L),\ \eta_U(L,\mu),\ \mathfrak h_U(L,\mu),\
			\mathcal V_U(L),\ \mathcal F_{\mathrm{th},U}[p_U],\
			\mathcal F_U^\sharp,\ S_U(L,\mu)<\infty.
			\label{eq:finite-error-quantities}
		\end{align}
		Consequently $\mathcal E_U(L,\mu)<\infty$.  For a family with varying $U$,
		these quantities are uniformly bounded only under the uniform bounded-geometry,
		support, weighted-radiation and state-seminorm hypotheses stated in the
		preceding results.
	\end{lemma}
	
	\begin{proof}
		Lemma~\ref{lem:coefficient-difference} and the compactness observation after
		\eqref{eq:window-kappa-variation} give finiteness of $C_A$, $a_{U,s}$ and
		$b_U$.  Definition~\eqref{eq:GammaF-constant} and the smooth-dependence
		argument following it give finite $C_{\Gamma,F;U}(L,\mu)$ for fixed
		$L<\infty$ and $\mu>0$.
		Proposition~\ref{prop:Duhamel} gives $G_U(L)<\infty$, while
		Lemma~\ref{lem:detector-reconstruction-constants} gives finite $C_D$ and
		$C_R$.  The Schwarzschild trace is in the order-$s$ radiation space, so
		\eqref{eq:exact-Sch-tail} is finite.  Assumption
		\ref{ass:state-energy-control} supplies finite $C_{\omega,U}$ and, because it
		includes the frozen outgoing channel, finite $\mathcal V_U(L)$.  Hence
		$R_U(L)<\infty$.  Both data entering \eqref{eq:bridge-defect} belong to
		$X^s$, so $\eta_U(L,\mu)<\infty$.  Proposition
		\ref{lem:finite-Hadamard} gives finite $\mathfrak h_U$, and the thermal
		quadratic form is finite on the compact profile $p_U$.  The remaining
		assertions follow from the definitions of $\mathcal F_U^\sharp$, $S_U$ and
		$\mathcal E_U$ by arithmetic operations on nonnegative finite numbers.
	\end{proof}
	
	\begin{theorem}[Certified thermal response interval]
		\label{thm:certified-response}
		Under Theorem~\ref{thm:target} and Assumption
		\ref{ass:state-energy-control}, fix positive $L$ and $\mu$ and use
		\eqref{eq:Fsharp-S-definition} and \eqref{eq:certified-error}.  Then
		\begin{equation}
			\max\{0,\mathcal F_{\mathrm{th},U}[p_U]-\mathcal E_U(L,\mu)\}
			\leq\mathcal F_g[h_U]
			\leq\mathcal F_{\mathrm{th},U}[p_U]+\mathcal E_U(L,\mu).
			\label{eq:two-sided-final}
		\end{equation}
	\end{theorem}
	
	\begin{proof}
		Put $g_0=\Gamma_{\Sigma_{t_-(U)}}g^{\mathrm{loc}}_{U,L,\mu}$,
		$f_0=g_0+f^{\mathrm{out,fr}}_{U,L}$ and $r=f_U-f_0$.
		Theorem~\ref{thm:target} and \eqref{eq:bridge-defect} give
		$\|r\|_{X^s}\leq S_U(L,\mu)$.
		The time-slice identity gives
		$q_U(g_0,g_0)=W_2(\overline{g^{\mathrm{loc}}},g^{\mathrm{loc}})$, so
		Proposition~\ref{lem:finite-Hadamard} bounds it by
		$\mathcal F_U^\sharp$.
		By Cauchy--Schwarz for the positive form,
		\begin{equation}
			q_U(f_0,f_0)^{1/2}
			\leq q_U(g_0,g_0)^{1/2}
			+\mathcal V_U(L)^{1/2}
			\leq\sqrt{\mathcal F_U^\sharp}+\sqrt{\mathcal V_U(L)}.
		\end{equation}
		Lemma~\ref{lem:positive-comparison} therefore bounds the change from $f_0$ to
		$f_U$ by the first line of \eqref{eq:certified-error}.  A second application of
		Cauchy--Schwarz bounds the outgoing diagonal and local-horizon--outgoing cross term
		by $\mathcal V_U+2\sqrt{\mathcal F_U^\sharp\mathcal V_U}$.  Finally use
		Proposition~\ref{lem:finite-Hadamard}.  Adding the three bounds proves
		\eqref{eq:two-sided-final}.
	\end{proof}
	
	Thus ``approximately thermal'' means membership in the explicit closed
	interval \eqref{eq:two-sided-final}.  The interval is rigorous at finite
	parameters because the bridge and Hadamard defects are retained as exact
	nonnegative quantities.  Convergence to the thermal reference follows when
	these quantities and the PDE and state terms tend to zero.
	
	If the ray-tracing map from an affine null coordinate $P$ in the past to Bondi
	retarded time $U$ is known, define
	\begin{equation}
		\kappa_{\mathrm{rt}}(U)
		=-\frac{P''(U)}{P'(U)}.
		\label{eq:peeling-function}
	\end{equation}
	An exactly exponential $P(U)$ gives an exact thermal ray-tracing kernel.  If
	the localisation bridge and the other errors in
	\eqref{eq:certified-error} vanish under controlled variation, the leading
	local temperature is
	\begin{equation}
		T_{\mathrm{eff}}(U)=\frac{\kappa_{\mathrm{rt}}(U)}{2\pi},
		\label{eq:effective-temperature}
	\end{equation}
	with the conditional meaning supplied by \eqref{eq:two-sided-final}, rather than by an
	unquantified adiabatic assertion.  The relevance of the peeling function, and
	its distinction from an arbitrary dynamical definition of surface gravity,
	was emphasised in \cite{BarceloLiberatiSonegoVisser2011}.
	
	When $m$ is constant, $a_{U,s}=v_U=b_U=0$,
	$\kappa_{\mathrm{rt}}=\kappa_{\mathrm{lin}}=\kappa_{\mathrm K}=1/(4M)$,
	frequency mixing becomes
	diagonal, and the desired result reduces to the Fredenhagen--Haag response
	\begin{equation}
		\mathcal F_{\mathrm{FH}}[h]
		=\mathrm{const.}
		\sum_{\ell,m}\int_{-\infty}^{+\infty}
		|D_\ell(\omega)|^2
		\frac{|\widetilde h_{\ell m}(\omega,-\omega)|^2}
		{\omega(1-e^{-8\pi M\omega})}\,\dd\omega.
		\label{eq:FH-final}
	\end{equation}
	Indeed, in the normalization of \cite{FredenhagenHaag1990} the Fourier
	coefficient of the horizon profile is
	\begin{equation}
		a_{\ell m}(\omega)
		=D_\ell(\omega)
		\frac{\widetilde h_{\ell m}(\omega,-\omega)}{2i\omega}.
	\end{equation}
	Inserting this into the Fourier form \eqref{eq:KMS-Fourier}, whose density is
	proportional to
	$\omega/(1-e^{-8\pi M\omega})$, gives \eqref{eq:FH-final}; the factors
	$r_*^2$, $2$ and the Fourier convention are contained in the displayed
	overall constant.  This establishes the convention matching between
	\eqref{eq:thermal-reference-definition} and the stationary
	Fredenhagen--Haag form.
	
	\section{Consequences for the admissible mass profiles}
	\label{sec:mass-profile-consequences}
	
	The finite-window comparison and the certified detector bound can now be
	applied to the mass profiles introduced in Section~\ref{sec:geometry}.  The
	resulting statements distinguish finite-time estimates, conditional
	late-time limits and the dependence on the global completion of the
	space-time.
	
	\subsection{The spherical massless sector on a compact slab}
	
	For $\ell=0$ and $\mu=0$, equation \eqref{eq:reduced-Vaidya} is
	\begin{equation}
		2\epsilon\psi_{wr}+\partial_r(C\psi_r)-\frac{2m(w)}{r^3}\psi=0.
	\end{equation}
	The multiplier estimate \eqref{eq:first-energy-estimate}, the explicit
	coefficient calculation \eqref{eq:explicit-coefficient-difference}, and
	Proposition~\ref{prop:Duhamel} prove finite-slab well-posedness and the frozen
	comparison.  No WKB or geometric-optics approximation is used.  Thus the
	spherical massless sector satisfies the finite-window propagation estimate,
	and its detector response is governed by the certified interval of
	Theorem~\ref{thm:certified-response}.
	
	\subsection{All angular momenta}
	
	Spherical symmetry gives $[P,\Omega_i]=0$ for the three rotation fields.
	Commuting with them and summing the energy estimates controls every angular
	Sobolev derivative.  More importantly, the term
	$\ell(\ell+1)/r^2$ cancels in \eqref{eq:explicit-coefficient-difference}; the
	time-dependent perturbation constant is uniform in $\ell$.  The Schwarzschild
	radiation-field isometry of Proposition~\ref{prop:stationary-tail} includes
	all modes; the weighted moment is imposed after commuting with rotations.
	Coudray's
	boundary energy \eqref{eq:conformal-energy-density} likewise contains the
	full angular gradient.  Consequently the finite-window comparison extends to
	all angular momenta with perturbation constants uniform in $\ell$.  The
	horizon-localisation contribution remains explicitly present in the error
	term of Theorem~\ref{thm:certified-response}.
	
	\subsection{Asymptotically stationary accretion}
	
	\begin{corollary}[Late-time accretion]
		\label{cor:accretion}
		Suppose $m(U)\to M_+>0$ and the selected null separatrix satisfies
		$r_H(U)\to2M_+$.  Choose spacelike slabs $D_{U,L(U)}$ with
		$L(U)\to\infty$ whose earliest sampled null coordinate tends to the future,
		\begin{equation}
			\underline w_U:=\inf_{x\in D_{U,L(U)}}w(x)\longrightarrow+\infty.
		\end{equation}
		Assume the constants in Theorem~\ref{thm:certified-response} are uniform on
		these slabs, the frozen profiles have the stationary normalization
		\eqref{eq:thermal-reference-definition}, and
		\begin{align}
			&G_U(L(U))a_{U,s}(L(U))\longrightarrow0,
			&&\varepsilon_U^{\mathrm{Sch}}(L(U))\longrightarrow0,\notag\\
			&\mathcal V_U(L(U))\longrightarrow0,
			&&\eta_U(L(U),\mu(U))+\mathfrak h_U(L(U),\mu(U))\longrightarrow0
			\label{eq:accretion-vanishing-assumptions}
		\end{align}
		for some $\mu(U)\downarrow0$.  Assume also uniform continuity in $M$ of the
		frozen Fredenhagen--Haag form on the detector family.  Then
		\begin{equation}
			\mathcal E_U(L(U),\mu(U))\to0.
		\end{equation}
		Consequently
		\begin{equation}
			\mathcal F_g[h_U]-\mathcal F_{\mathrm{FH},M_+}[h_U]\longrightarrow0.
			\label{eq:accretion-limit}
		\end{equation}
	\end{corollary}
	
	\begin{proof}
		All terms in \eqref{eq:certified-error} vanish by virtue of \eqref{eq:accretion-vanishing-assumptions} together with the uniform state constant. The stationary normalization identifies the reference form with \(\mathcal{F}_{\mathrm{FH},m_U}[h_U]\), and the assumed continuity of this form as \(m_U\to M_+\) then establishes \eqref{eq:accretion-limit}.
		
		Here is a concrete sufficient condition for the first line of
		\eqref{eq:accretion-vanishing-assumptions}.  Suppose, in addition, that
		\begin{equation}
			d_s(T)=\sup_{w\geq T}\left(
			\left|\frac{m(w)}{M_+}-1\right|
			+\sum_{j=1}^{s+2}M_+^{j-1}|\partial_w^jm(w)|
			+M_+\left|\kappa_{\mathrm{lin}}(w)-(4M_+)^{-1}\right|\right)
		\end{equation}
		decreases to zero, that the chosen slabs obey
		$\underline w_U\geq U/2$ and $\Delta v_U(L)\leq C_vL$, and that
		$G_U(L)$ is uniformly bounded.  An explicit admissible choice is
		\begin{equation}
			L(U)=M_+\min\left\{\sqrt{U/M_+},\ d_s(U/2)^{-1/2}\right\},
			\label{eq:accretion-window-choice}
		\end{equation}
		with the second entry omitted when $d_s(U/2)=0$.  Then $L(U)\to\infty$ and
		$L(U)d_s(U/2)/M_+\to0$.  On these slabs the definitions give, with constants
		depending only on the bounded coefficient class,
		\begin{equation}
			a_{U,s}(L)\leq C M_+L d_s(U/2),\qquad
			b_U(L)\leq C\frac{\Delta v_U(L)}{M_+}d_s(U/2),
		\end{equation}
		so the coefficient and peeling mismatches vanish.  A uniform weighted
		radiation norm in \eqref{eq:Sch-moment-tail} gives
		$\varepsilon_U^{\mathrm{Sch}}(L)\to0$ because
		$\kappa_UL\to\infty$.  The remaining bridge, KPV and state-channel limits
		are additional assumptions; the derivative envelope does not
		prove them.
	\end{proof}
	
	Corollary~\ref{cor:accretion} therefore identifies sufficient quantitative
	conditions for convergence to the Fredenhagen--Haag response.  Profiles in
	$\mathfrak V_{\mathrm S}$ are the special case in which $d_s(T)$ vanishes
	identically after a finite time.
	
	\subsection{Evaporation followed by accretion}
	
	The class $\mathfrak V_{\mathrm T}$ permits the same finite-window analysis;
	neither Lemma~\ref{lem:coefficient-difference} nor Proposition~\ref{prop:Duhamel}
	uses a fixed sign of $m'$.  If the null-coordinate range $[a,b]$ sampled by
	a slab crosses the turnaround $w_b$, the one-dimensional total variation of
	the mass is exactly
	\begin{equation}
		\operatorname{Var}_{[a,b]}m=m(a)-m_b+m(b)-m_b.
		\label{eq:turnaround-total-variation}
	\end{equation}
	The slice integral $v_U(L)$ in \eqref{eq:window-variations} is not literally
	this variation because the spacelike slices sample ranges of $w$; on a fixed
	regular foliation it is bounded in terms of the same absolute derivative and
	the foliation constants.  Higher coefficient seminorms in $a_{U,s}$ must be
	controlled separately.  Thus the evaporating and accreting portions enter
	energy estimates through absolute variations rather than cancelling.  A
	return to the initial mass does not undo frequency mixing produced during the
	first phase.
	
	\begin{proposition}[Finite-window and late-time turnaround statements]
		\label{prop:turnaround}
		Let $m\in\mathfrak V_{\mathrm T}$ and assume the causal and analytic
		hypotheses of Theorems~\ref{thm:target} and
		\ref{thm:certified-response}.  Then every compact observation window with
		$m\geq m_b>0$ obeys the certified response interval
		\eqref{eq:two-sided-final}; when its sampled null-coordinate range crosses
		$w_b$, the mass variation is given by
		\eqref{eq:turnaround-total-variation}.  If in
		addition $m(w)\to M_+>0$, the peeling, state and uniformity hypotheses of
		Corollary~\ref{cor:accretion} hold on the future tail, and the chosen windows
		eventually lie to the future of $w_b$, then
		\begin{equation}
			\mathcal F_g[h_U]-\mathcal F_{\mathrm{FH},M_+}[h_U]\longrightarrow0.
			\label{eq:turnaround-late-limit}
		\end{equation}
	\end{proposition}
	
	\begin{proof}
		The coefficient and Duhamel estimates contain absolute coefficient
		seminorms, so they remain valid across a change of sign of $m'$.  Monotonicity on each side of $w_b$
		gives \eqref{eq:turnaround-total-variation}.  The first assertion is then
		Theorem~\ref{thm:certified-response}.  For late windows lying wholly after
		$w_b$, the profile is monotone increasing and Corollary~\ref{cor:accretion}
		applies, proving \eqref{eq:turnaround-late-limit}.
	\end{proof}
	
	The transient response on windows meeting the first phase retains its effect
	through the dynamical scattering kernel and the outgoing-state term
	$\mathcal V_U$.  The late thermal limit \eqref{eq:turnaround-late-limit}, when
	its hypotheses hold, says only that those transients disperse from the chosen
	detector channel; it does not say that the intermediate evaporation was
	physically absent.
	
	\subsection{Finite-slab evaporation and the extension obstruction}
	
	If a decreasing finite slab is completed to a future Schwarzschild exterior
	of mass $M_+>0$, Theorem~\ref{thm:Coudray-scattering} supplies the exact
	boundary scattering map in the orientation covered there, and the stationary
	future version of Corollary~\ref{cor:accretion} gives the late detector limit.
	For a black-hole interpretation one must use a globally consistent matched
	geometry, as explained after \eqref{eq:Vaidya-stress}.
	
	\begin{proposition}[A finite slab has no intrinsic late-time response]
		\label{prop:extension-obstruction}
		The geometry and state on a compact Vaidya slab do not determine a unique detector response as \(U\to\infty\).
	\end{proposition}
	
	\begin{proof}
		Choose two globally hyperbolic future extensions which agree on an open
		neighbourhood of the closed slab, together with the restriction of the field
		state there, but which later settle to Schwarzschild masses
		$M_+^{(1)}\neq M_+^{(2)}$ after different exterior fluxes are supplied.
		Their complete Einstein--matter Cauchy data outside the common region are not
		being held fixed; such data were never part of the finite-slab hypothesis.
		Finite propagation speed makes all observables in the common slab identical.
		Choose the construction so that the common past contains a Cauchy surface
		$\Sigma$ for the test-field problem and prescribe the same Hadamard Cauchy
		covariance there.  Transport by the two future Klein--Gordon operators gives
		two future covariances.  They cannot agree on every future smearing: by
		polarisation, equality of all their quadratic forms would force equality of
		their antisymmetric parts $iE_j$, whereas the causal propagators differ for
		some test function whose causal hull meets the region where the metrics
		differ.  Thus some future detector has different responses in the two
		extensions.  If each completion
		is additionally supplied with its corresponding future-Unruh state and the
		hypotheses of Corollary~\ref{cor:accretion}, the two limiting Planck factors
		have inverse temperatures $8\pi M_+^{(1)}$ and $8\pi M_+^{(2)}$.  Therefore no
		theorem whose assumptions mention only the common finite slab can select one
		late-time answer.
	\end{proof}
	
	Proposition~\ref{prop:extension-obstruction} therefore separates the two
	relevant conclusions.  The certified finite-window bound depends only on the
	corresponding causal development, whereas a late-time limit requires a
	specified future completion.
	
	\subsection{Weighted estimates for asymptotic evaporation}
	\label{sec:asymptotic-evaporation}
	
	For the outgoing metric put
	\begin{equation}
		\tau(u)=\int_{u_0}^{u}\frac{\dd s}{m(s)},
		\qquad \rho=\frac r{m(u)}.
		\label{eq:scaled-coordinates}
	\end{equation}
	A direct substitution gives the exact conformal factorisation
	\begin{equation}
		g=m(u)^2\widehat g,
		\qquad
		\widehat g=
		-\left(1-\frac2\rho+2\rho m'(u)\right)\dd\tau^2
		-2\,\dd\tau\dd\rho+\rho^2\dd\Omega^2.
		\label{eq:rescaled-evaporation-metric}
	\end{equation}
	For the massless conformal equation the field rescaling converts the physical
	problem exactly to the conformal operator
	$\Box_{\widehat g}-R_{\widehat g}/6$, not in general to the bare wave
	operator $\Box_{\widehat g}$.  Thus fixed compact
	$\rho$ regions have uniform geometry whenever the dimensionless quantities
	$m'$, $mm''$, and their higher $\tau$ derivatives are bounded.
	
	The coordinate $\tau$ is null and must not be used as an energy time.  Put
	\begin{equation}
		T=\tau+\rho.
		\label{eq:rescaled-temporal-time}
	\end{equation}
	The inverse of the $(\tau,\rho)$ part of \eqref{eq:rescaled-evaporation-metric}
	gives
	\begin{equation}
		\widehat g^{-1}(\dd T,\dd T)
		=-1-\frac2\rho+2\rho m'(u)<0
		\label{eq:rescaled-temporal-norm}
	\end{equation}
	for $m'(u)\leq0$.  Thus $T$ is a temporal function on the evaporating region,
	uniformly so on every compact positive $\rho$-annulus.
	
	\begin{proposition}[Scale-covariant finite-window estimate]
		\label{prop:scaled-window}
		Let $q(U)>0$ and consider a $T$-window of length $q(U)$ in a fixed compact
		positive $\rho$-annulus whose causal development remains in the annulus.
		Assume the unit-mass Schwarzschild reference evolution is uniformly bounded
		on the chosen rescaled energy space (as it is for the standard global
		Schwarzschild energy space).  Define
		\begin{equation}
			\widehat\delta_U(q)=
			\sup\left\{|m'|+|mm''|+\cdots+|\partial_\tau^{N-1}m'|\right\}
			\label{eq:rescaled-variation}
		\end{equation}
		on that window.  For detector profiles and state norms supported in the
		annulus and uniformly bounded in the rescaled energy spaces, the
		propagation part of the response error satisfies
		\begin{equation}
			\widehat B_U(q)
			\leq Cq e^{Cq\widehat\delta_U(q)}\widehat\delta_U(q)
			+\varepsilon^{\mathrm{Sch}}_{1}(q).
			\label{eq:weighted-evaporation-bound}
		\end{equation}
		Hence if $q(U)\to\infty$ and
		$q(U)\widehat\delta_U(q(U))\to0$, the scale-covariant frozen response error
		tends to zero whenever the unit-mass Schwarzschild radiation tail tends to
		zero.  Under the weighted moment hypothesis \eqref{eq:Sch-moment-tail} that
		tail is $O(q^{-j})$.  For example, if an a priori envelope
		$\overline\delta_U\to0$ bounds \eqref{eq:rescaled-variation} throughout the
		window of length $\overline\delta_U^{-1/2}$, taking
		$q(U)=\overline\delta_U^{-1/2}$ gives an
		$O(\overline\delta_U^{1/2})+O(\overline\delta_U^{j/2})$ bound.
	\end{proposition}
	
	\begin{proof}
		Equation~\eqref{eq:rescaled-temporal-norm} and
		Proposition~\ref{prop:compact-evolution} give finite-window commuted energy
		estimates on the fixed annulus.  In \eqref{eq:rescaled-evaporation-metric} the difference
		from the unit-mass Schwarzschild coefficients, the conformal curvature
		potential, and all commuted differences through the required order are
		bounded by $C\widehat\delta_U(q)$.  Duhamel's formula relative to the assumed
		uniformly bounded Schwarzschild evolution, followed by Gronwall on an
		interval of length $q$, gives the first term in
		\eqref{eq:weighted-evaporation-bound}.  The exact
		Schwarzschild radiation decomposition, Proposition~\ref{prop:stationary-tail},
		gives the second.  The stated choices of $q$ make both terms vanish.
	\end{proof}
	
	Restoring physical units, the frozen reference in
	Proposition~\ref{prop:scaled-window} has inverse temperature
	$\beta_U=8\pi m(U)$.  The PDE proposition controls comparison with that
	reference scale; it does not alone establish a measured temperature.  After
	the bridge and state hypotheses are added, the associated scale-following
	temperature diverges as the horizon shrinks, but this still does not assert
	that a fixed far-away detector receives an unlimited thermal flux.
	
	There is nevertheless no ordinary short-range scattering limit at
	$u=+\infty$.  For every monotone $m(u)\downarrow0$,
	\begin{equation}
		\int_{\tau(u_0)}^{\infty}|m'(u(\tau))|\,\dd\tau
		=\int_{u_0}^{\infty}\frac{|m'(u)|}{m(u)}\,\dd u
		=\lim_{u\to\infty}\log\frac{m(u_0)}{m(u)}=+\infty.
		\label{eq:long-range-obstruction}
	\end{equation}
	Thus the rescaled perturbation is inevitably long-range even when $m'(u)$
	tends to zero.  Proposition~\ref{prop:scaled-window} proves a local-in-scale
	adiabatic theorem, not a global wave-operator theorem.  It also concerns
	detectors whose profiles are fixed relative to the shrinking radius; a
	detector at fixed physical radius requires a separate far-zone matching
	estimate.  The rescaled PDE comparison therefore applies to scale-following
	profiles on compact $\rho$-annuli.  A thermal detector theorem
	uses the localisation and state inputs above, and a detector at fixed
	physical radius requires a separate far-zone estimate.
	
	By \eqref{eq:horizon-curvature}, the physical
	horizon curvature diverges as $m\to0$.  Conformal covariance controls the
	massless test field but does not make semiclassical backreaction uniformly
	small.  Moreover the usual law $m'=-\alpha/m^2$ reaches zero in finite time;
	an asymptotic profile necessarily modifies that law at small mass.
	
	\section{Conclusions}
	
	We have developed a finite-window extension of the
	Fredenhagen--Haag strategy for massless scalar fields on controlled Vaidya
	backgrounds.  Regular detector--horizon causal windows admit globally
	hyperbolic developments, and the corresponding nonautonomous wave evolution
	can be compared with a frozen Schwarzschild propagator on explicit Sobolev
	energy spaces.  Theorem~\ref{thm:target} quantifies this comparison in terms
	of coefficient variation and a stationary radiation tail.  The null-ray
	analysis supplies the exact dynamical peeling coefficient together with a
	quadratic remainder estimate.  Combining these results with the local
	Hadamard scaling distribution and positivity yields the certified detector
	interval of Theorem~\ref{thm:certified-response}.  In particular, the
	statement at finite parameters is a two-sided inequality with explicitly
	defined error terms rather than an uncontrolled thermal approximation.
	
	The applications distinguish the late-time behaviour permitted by the
	different mass profiles.  Under the quantitative decay assumptions of
	Corollary~\ref{cor:accretion}, asymptotically stationary accretion converges
	to the Fredenhagen--Haag response of the limiting Schwarzschild geometry.
	The same conclusion holds for evaporation--accretion turnaround profiles
	when the observation windows enter the stationary future regime.  A finite
	evaporating slab determines the finite-window response but not a unique
	late-time limit until its future extension is prescribed.  For asymptotic
	evaporation with $m(u)>0$ and $m(u)\to0$, the conformally rescaled problem
	satisfies the scale-covariant estimate
	\eqref{eq:weighted-evaporation-bound}, while
	\eqref{eq:long-range-obstruction} proves that the global rescaled
	perturbation is necessarily long-range.
	
	The asymptotic detector limits require the decay of the localisation defect
	$\eta_U$, together with the state-channel and uniformity conditions stated in
	the corresponding corollaries.  A fixed-radius detector in the asymptotically
	evaporating geometry additionally requires a far-zone propagation estimate.
	The one-stream Vaidya model also does not incorporate the self-consistent
	backreaction of simultaneous ingoing and outgoing fluxes.  These restrictions
	do not affect the finite-window estimates or the certified detector
	inequality proved here; they specify the additional hypotheses needed when
	those results are promoted to global late-time statements.


\begin{thebibliography}{99}
		
		\bibitem{FredenhagenHaag1990}
		K.~Fredenhagen and R.~Haag,
		``On the derivation of Hawking radiation associated with the formation of a
		black hole,''
		\emph{Commun. Math. Phys.} \textbf{127} (1990), 273--284.
		\href{https://doi.org/10.1007/BF02096757}{doi:10.1007/BF02096757}.
		
		\bibitem{DimockKay1987}
		J.~Dimock and B.~S.~Kay,
		``Classical and quantum scattering theory for linear scalar fields on the
		Schwarzschild metric. I,''
		\emph{Ann. Phys.} \textbf{175} (1987), 366--426.
		\href{https://doi.org/10.1016/0003-4916(87)90214-4}
		{doi:10.1016/0003-4916(87)90214-4}.
		
		\bibitem{DafermosRodnianski2013}
		M.~Dafermos and I.~Rodnianski,
		``Lectures on black holes and linear waves,'' in
		\emph{Evolution Equations}, Clay Mathematics Proceedings, vol.~17,
		American Mathematical Society, 2013, pp.~97--205.
		\href{https://arxiv.org/abs/0811.0354}{arXiv:0811.0354}.
		
		\bibitem{BaskinWang2014}
		D.~Baskin and F.~Wang,
		``Radiation fields on Schwarzschild spacetime,''
		\emph{Commun. Math. Phys.} \textbf{331} (2014), 477--506.
		\href{https://arxiv.org/abs/1305.5273}{arXiv:1305.5273}.
		
		\bibitem{Nicolas2016}
		J.-P.~Nicolas,
		``Conformal scattering on the Schwarzschild metric,''
		\emph{Ann. Inst. Fourier} \textbf{66} (2016), 1175--1216.
		\href{https://arxiv.org/abs/1312.1386}{arXiv:1312.1386}.
		
		\bibitem{DappiaggiMorettiPinamonti2011}
		C.~Dappiaggi, V.~Moretti and N.~Pinamonti,
		``Rigorous construction and Hadamard property of the Unruh state in
		Schwarzschild spacetime,''
		\emph{Adv. Theor. Math. Phys.} \textbf{15} (2011), 355--447.
		\href{https://arxiv.org/abs/0907.1034}{arXiv:0907.1034}.
		
		\bibitem{Coudray2024}
		A.~Coudray,
		``Conformal scattering of the wave equation in the Vaidya spacetime,''
		\emph{Rev. Math. Phys.}, online publication (2025), 2550035.
		\href{https://doi.org/10.1142/S0129055X25500357}%
		{doi:10.1142/S0129055X25500357};
		\href{https://arxiv.org/abs/2405.08659}{arXiv:2405.08659}.
		
		\bibitem{ChirentiSaa2012}
		C.~Chirenti and A.~Saa,
		``Double-null formulation of the general Vaidya metric,''
		\emph{Class. Quantum Grav.} \textbf{29} (2012), 135003.
		\href{https://doi.org/10.1088/0264-9381/29/13/135003}%
		{doi:10.1088/0264-9381/29/13/135003}.
		
		\bibitem{KurpiczPinamontiVerch2021}
		F.~Kurpicz, N.~Pinamonti and R.~Verch,
		``Temperature and entropy--area relation of quantum matter near spherically
		symmetric outer trapping horizons,''
		\emph{Lett. Math. Phys.} \textbf{111} (2021), 110.
		\href{https://doi.org/10.1007/s11005-021-01445-7}{doi:10.1007/s11005-021-01445-7}.
		
		\bibitem{Kodama1980}
		H.~Kodama,
		``Conserved energy flux for the spherically symmetric system and the back
		reaction problem in the black hole evaporation,''
		\emph{Prog. Theor. Phys.} \textbf{63} (1980), 1217--1228.
		\href{https://doi.org/10.1143/PTP.63.1217}{doi:10.1143/PTP.63.1217}.
		
		\bibitem{Hayward1998}
		S.~A.~Hayward,
		``Unified first law of black-hole dynamics and relativistic thermodynamics,''
		\emph{Class. Quantum Grav.} \textbf{15} (1998), 3147--3162.
		\href{https://doi.org/10.1088/0264-9381/15/10/017}{doi:10.1088/0264-9381/15/10/017}.
		
		\bibitem{JanssenVerch2023}
		D.~W.~Janssen and R.~Verch,
		``Hadamard states on spherically symmetric characteristic surfaces, the
		semi-classical Einstein equations and the Hawking effect,''
		\emph{Class. Quantum Grav.} \textbf{40} (2023), 045002.
		\href{https://doi.org/10.1088/1361-6382/acb039}%
		{doi:10.1088/1361-6382/acb039}.
		
		\bibitem{BaerGinouxPfaeffle2007}
		C.~B\"ar, N.~Ginoux and F.~Pf\"affle,
		\emph{Wave Equations on Lorentzian Manifolds and Quantization},
		ESI Lectures in Mathematics and Physics, European Mathematical Society, 2007.
		\href{https://arxiv.org/abs/0806.1036}{arXiv:0806.1036}.
		
		\bibitem{Radzikowski1996}
		M.~J.~Radzikowski,
		``Micro-local approach to the Hadamard condition in quantum field theory on
		curved space-time,''
		\emph{Commun. Math. Phys.} \textbf{179} (1996), 529--553.
		
		\bibitem{Hiscock1981}
		W.~A.~Hiscock,
		``Models of evaporating black holes. I,''
		\emph{Phys. Rev. D} \textbf{23} (1981), 2813--2822.
		\href{https://doi.org/10.1103/PhysRevD.23.2813}{doi:10.1103/PhysRevD.23.2813}.
		
		\bibitem{Hiscock1981II}
		W.~A.~Hiscock,
		``Models of evaporating black holes. II. Effects of the outgoing created
		radiation,''
		\emph{Phys. Rev. D} \textbf{23} (1981), 2823--2827.
		\href{https://doi.org/10.1103/PhysRevD.23.2823}{doi:10.1103/PhysRevD.23.2823}.
		
		\bibitem{DwivediJoshi1989}
		I.~H.~Dwivedi and P.~S.~Joshi,
		``On the nature of naked singularities in Vaidya spacetimes,''
		\emph{Class. Quantum Grav.} \textbf{6} (1989), 1599--1606.
		\href{https://doi.org/10.1088/0264-9381/6/11/013}%
		{doi:10.1088/0264-9381/6/11/013}.
		
		\bibitem{PegasEtAl2025}
		J.~V.~O.~P\^egas, A.~G.~S.~Landulfo, G.~E.~A.~Matsas and D.~A.~T.~Vanzella,
		``Globally hyperbolic evaporating black hole and the information loss issue,''
		\emph{Class. Quantum Grav.} \textbf{42} (2025), 065009.
		\href{https://doi.org/10.1088/1361-6382/adb534}%
		{doi:10.1088/1361-6382/adb534}.
		
		\bibitem{BarceloLiberatiSonegoVisser2011}
		C.~Barcel\'o, S.~Liberati, S.~Sonego and M.~Visser,
		``Minimal conditions for the existence of a Hawking-like flux,''
		\emph{Phys. Rev. D} \textbf{83} (2011), 041501.
		\href{https://doi.org/10.1103/PhysRevD.83.041501}{doi:10.1103/PhysRevD.83.041501}.
		
	\end{thebibliography}
\end{document}